\documentclass[11pt,letterpaper]{amsart}
\usepackage{hyperref}
\usepackage{joshuag}
\usepackage{joshuag-layout}
\usepackage{graphicx}

\newcommand{\PP}{\ensuremath{\mathbb{P}}}

\DeclareMathOperator{\rk}{rk}
\DeclareMathOperator{\Hess}{Hess}
\newtheorem{fact}[defn]{Fact}
\newtheorem{obs}[defn]{Observation}

\DeclareMathOperator{\poly}{poly}

\DeclareMathOperator{\Poly}{Poly}
\DeclareMathOperator{\Coeff}{Coeff}

\DeclareMathOperator{\GL}{GL}
\DeclareMathOperator{\AGL}{AGL}

\newcommand{\var}[1]{\mathbf{#1}}
\renewcommand{\det}{\mathop{\mathrm{det}}\nolimits}

\begin{document}

\title{Unifying and Generalizing Known Lower Bounds \\
via Geometric Complexity Theory
}

\author{Joshua A. Grochow}
\thanks{Department of Computer Science, The University of Toronto. \tt{jgrochow@cs.toronto.edu}}

\begin{abstract}
We show that most arithmetic circuit lower bounds and relations between lower bounds naturally fit into the representation-theoretic framework suggested by geometric complexity theory (GCT), including: the partial derivatives technique (Nisan--Wigderson), the results of Razborov and Smolensky on $\cc{AC}^0[p]$, multilinear formula and circuit size lower bounds (Raz \etal), the degree bound (Strassen, Baur--Strassen), the connected components technique (Ben-Or), depth 3 arithmetic circuit lower bounds over finite fields (Grigoriev--Karpinski), lower bounds on permanent versus determinant (Mignon--Ressayre, Landsberg--Manivel--Ressayre), lower bounds on matrix multiplication (B\"{u}rgisser--Ikenmeyer) (these last two were already known to fit into GCT), the chasms at depth 3 and 4 (Gupta--Kayal--Kamath--Saptharishi; Agrawal--Vinay; Koiran), matrix rigidity (Valiant) and others. That is, the original proofs, with what is often just a little extra work, already provide representation-theoretic obstructions in the sense of GCT for their respective lower bounds. This enables us to expose a new viewpoint on GCT, whereby it is a natural unification and broad generalization of known results. It also shows that the framework of GCT is at least as powerful as known methods, and gives many new proofs-of-concept that GCT can indeed provide significant asymptotic lower bounds. This new viewpoint also opens up the possibility of fruitful two-way interactions between previous results and the new methods of GCT; we provide several concrete suggestions of such interactions. For example, the representation-theoretic viewpoint of GCT naturally provides new properties to consider in the search for new lower bounds.
\end{abstract}

\maketitle

\pagestyle{myheadings}
\markboth{Unifying and Generalizing Known Lower Bounds via GCT - Joshua A. Grochow}{Unifying and Generalizing Known Lower Bounds via GCT - Joshua A. Grochow}

\section{Introduction}
Geometric complexity theory (GCT) is a program towards lower bounds---such as $\cc{P} \neq \cc{NP}$---using algebraic geometry and representation theory (see \cite{gctJACM} for an overview, and references therein). In this paper, we show that most arithmetic circuit lower bounds naturally fit into the representation-theoretic framework used in GCT. We also show that part of the representation-theoretic approach is necessary, that this approach illuminates lower bounds even when it is not strictly necessary, and that it may in fact be the easiest approach to proving circuit lower bounds. GCT thus provides a unifying and generalizing framework for many known lower bounds. This representation-theoretic viewpoint opens the door for new potentially fruitful two-way interactions between previous results and new progress in (geometric) complexity theory (see Sections~\ref{sec:implications} and \ref{sec:old} for details).

This paper presupposes no knowledge of representation theory on the part of the reader. In fact, we use previous lower bounds together with our new viewpoint to motivate the use and definitions of representation theory and algebraic geometry in complexity theory.

Essentially any lower bound proof $\mathcal{C}_{hard} \not\subseteq \mathcal{C}_{easy}$ between nonuniform complexity classes proceeds by finding some ``useful'' property, which applies to every function in $\mathcal{C}_{easy}$, but not to every function in $\mathcal{C}_{hard}$. The first part of the GCT program suggests the use of properties of a certain type, namely (linear-)invariant properties defined by the vanishing of polynomials, which we capture in the notion of ``separating module'' (Definition~\ref{def:sep}). Recall that a property $\Pi$ is linear-invariant if for every function on $n$ variables, $f(\var{x})$ has $\Pi$ if and only if $f(A\var{x})$ has $\Pi$ for every invertible $n \times n$ change of variables $A$. In this paper we show that most known arithmetic circuit lower bounds in fact use separating modules, including:
\begin{itemize}
\item Lower bounds on restricted depth 3 arithmetic circuits in
characteristic zero \cite{nisanWigderson}

\item Lower bounds on (unrestricted) depth 3 arithmetic circuits over
finite fields \cite{GK}

\item The recent lower bounds on depth 4 arithmetic circuits with
bottom fan-in $O(\sqrt{n})$ \cite{GKKS}

\item Lower bounds on multilinear formula size \cite{razPerm}

\item The degree bound of Strassen \cite{strassenDegree} and Baur--Strassen \cite{baurStrassen} (see below)

\item Lower bounds on real (semi-)algebraic decision trees \cite{benOr,yao97}

\item Lower bounds on bounded depth Boolean circuits \cite{razborov,smolensky}

\item The best known lower bounds ($n^2/2$) on permanent versus
determinant \cite{mignonRessayre} (already shown to use a separating module \cite{LMR})

\item Many lower bounds on matrix multiplication (already shown to use a separating module \cite{BI, strassenMM})
\end{itemize}
We expect that results which use similar techniques can be shown to use separating modules as well, such as \cite{razFormula,RSY,razYehudayoffMultilinear,shpilkaWigderson,GR,yao91,BLY}.
We also observe that many relations between lower bounds yield relations between separating modules. In other words, if lower bound A is proved using a separating module, that yields a separating module for lower bound B:
\begin{itemize}
\item Lower bounds on partial derivatives implies lower bounds \cite{baurStrassen}

\item Matrix rigidity implies circuit lower bounds \cite{valiantRigidity}

\item The chasm at depth 4 \cite{AV, koiranChasm} and the recent chasm
at depth 3 \cite{GKKSchasm}

\item Tensor-rank lower bounds imply formula size lower bounds \cite{razTensor}
\end{itemize}

Finally, in Section~\ref{sec:necessary} we argue that the use of invariant properties is essentially necessary, and we give heuristic arguments that the use of separating modules is by far the easiest way to prove arithmetic circuit lower bounds. Thus separating modules are the first approach to try, and indeed may be the only approach that is easy enough that it will ever be carried out. We can already give one such heuristic argument: most arithmetic circuit lower bounds already use separating modules.

This new viewpoint makes new tools available, and suggests new conjectures and directions to better understand complexity classes and lower bounds. However, we do not provide new proofs of any of the above results. Our paper is similar in some ways to Natural Proofs \cite{razborovRudich} or Razborov's papers \cite{razborovBA1, razborovBA2} on bounded arithmetic, in that we offer a meta-observation about many lower bounds. This involves digging into the details of the proofs of known lower bounds to understand them in a particular way, which is sometimes trivial but sometimes requires new insights. These previous meta-results have shown that a new viewpoint can be quite fruitful; for example, by working in the framework of bounded arithmetic, Razborov was able to come up with a beautiful new proof of the Switching Lemma \cite{razborovBA1}. Despite this new proof of a lower bound against $\cc{AC}^0$, the fundamental message of the papers \cite{razborovRudich, razborovBA1, razborovBA2} was negative, giving barriers to proving strong lower bounds, whereas the message of this paper is \emph{positive}, suggesting a \emph{route to proving lower bounds}---a route that most arithmetic circuit lower bounds have already begun to traverse.

In Section~\ref{sec:implications}, we discuss some of the implications of this work. We postpone further details of the implications until Section~\ref{sec:discussion}, as they are difficult to discuss properly without definitions and a full example in mind. We give the definitions and an example of how a previous lower bound fits into this new viewpoint in Section~\ref{sec:ex0}. In Section~\ref{sec:necessary} we argue for the necessity of invariant properties and the feasibility and utility of separating modules, especially in comparison with other possible approaches. Section~\ref{sec:discussion} contains further discussion and implications. In particular, we discuss the relation of this viewpoint to the larger GCT program---in particular, separating modules are only the very beginning of the GCT approach. We also discuss lower bounds which don't seem to fit into this framework---mostly those based on uniform hierarchy theorems---and we suggest some concrete directions for future research to push forward both our understanding of GCT, and our understanding of known lower bounds and the complexity classes they consider. We also discuss in what way Boolean lower bounds fit into this framework. In Sections~\ref{sec:lowerbounds} and \ref{sec:relations} we prove that the results mentioned above use separating modules. However, if the reader is willing to take the above lists on faith, the significance of this paper can be understood without reading these last two sections in detail.

\subsection{Implications} \label{sec:implications}
Our unifying viewpoint suggests the possibility of a fruitful two-way interplay between the methods currently being leveraged in GCT against major open problems like permanent versus determinant and $\cc{P}$ versus $\cc{NP}$, and already hard-won knowledge for lower bounds on more tractable problems. Although we can state some of these possible interactions now, they will become clearer after the example in the next section, and we discuss further implications in Section~\ref{sec:discussion}.

First, the representation-theoretic viewpoint suggests where to look for new properties that might yield lower bounds. Even for lower bounds that are already essentially tight, the representation theory suggests how we might get new proofs of these lower bounds or otherwise understand them better. 

Second, the representation-theoretic viewpoint suggests new conjectures, directions, and techniques that may prove fruitful; see, for example, the last paragraph of Section~\ref{sec:gct} and the open questions in Sections~\ref{sec:old} and \ref{sec:explicit}.

Third, by showing that previous lower bounds and GCT share a common representation-theoretic viewpoint, we reveal many new contexts in which it might hopefully be easier to develop the tools and techniques of algebraic geometry and representation theory needed for the GCT approach to bigger problems such as permanent versus determinant or $\cc{P}$ versus $\cc{NP}$.

Fourth, it is often asked how difficult it is to re-prove known lower bounds using GCT. The viewpoint in this paper reveals that most of the old proofs \emph{already} give representation-theoretic knowledge crucial to the GCT approach, in the form of separating modules. There is, however, a difference between separating modules and the geometric obstructions defined in \cite{gct2}. Upgrading the previous lower bounds to yield such geometric obstructions is one of the open questions we discuss in more detail in Section~\ref{sec:gct}. This is one of the ways in which GCT suggests how we might understand previous lower bounds better, even ones that are essentially tight.

For now, we mention just one more point: the representation-theoretic viewpoint replaces the amorphous notion of ``useful property'' with the specific mathematical notion of separating module. In Section~\ref{sec:necessary} we argue that this is in some sense without loss of generality. This reduces an amorphous search for new useful properties to a comparatively feasible search for separating modules, which can even be made computational (see Appendix~\ref{app:necessary} and Section~\ref{sec:discussion} for more).

\section{Definitions and a motivating example} \label{sec:ex0}
Most nonuniform lower bounds $\mathcal{C}_{hard} \not\subseteq \mathcal{C}_{easy}$ are proved by finding a property shared by all functions in the ``easy'' class $\mathcal{C}_{easy}$ that some function $f \in \mathcal{C}_{hard}$ does not have. The goal of this section is to introduce a representation-theoretic formalization of the types of properties used by most arithmetic circuit lower bounds, namely (linear-)invariant properties defined by polynomials.

\subsection{Properties defined by polynomials}
Throughout the definitions and motivation, we will use the example of the space $\Poly^2(x,y) = \{a x^2 + b x y + c y^2 | a,b,c \in \F\}$ of degree 2 homogeneous polynomials in two variables\footnote{The notation $\Poly^d(x_1, \dotsc, x_n)$ is not standard. We use it because it is clear and mnemonic. For reference we give the standard notation from the literature in Appendix~\ref{app:notation}.}  over some field\footnote{In some of these examples, it may be necessary to restrict the characteristic of the field. In all of our actual results we specify the field more carefully.} $\F$, and the expression $b^2 - 4ac$. The space $\Poly^2(x,y)$ in this running example should be thought of as analogous to the space of polynomials we care about, like the determinant, permanent, etc. (which are points in $\Poly^n(x_{11},x_{12},\dotsc,x_{nn})$), but is small enough that we can carry out computations completely by hand and the definitions in this context should already be familiar to the reader. 

Recall that $b^2-4ac=0$ if and only if $ax^2 + b x y + c y^2$ is a perfect square\footnote{Equivalently and perhaps more familiar is that $b^2-4ac = 0$ if and only if $a x^2 + bx + c$ has a double root.} $(\alpha x + \beta y)^2$ for some constants $\alpha,\beta \in \F$. We thus view $b^2 - 4ac \iseq 0$ as a test for the property of being a perfect square, and we say that this \definedWord{property is defined by the (vanishing of the) polynomial} $b^2-4ac$. 

Note that here we consider $b^2-4ac$ not just as an expression, but as a polynomial in the \emph{variables} $a,b,c$, which are the coefficients of the polynomials $ax^2 + b xy + cy^2$. Because there are two different spaces of polynomials here, we find it useful to give different names to them. We refer to polynomials such as $ax^2 + bxy + cy^2 \in \Poly^2(x,y)$ with $a,b,c$ constants as \definedWord{input polynomials}: these are polynomials in the ``input variables'' $x,y$, and are also themselves inputs for the property tests. We refer to polynomials such as $b^2-4ac$ as \definedWord{test polynomials}: these are polynomials whose variables are the \emph{coefficients} of the input polynomials, and define a test for some property of input polynomials.

We index monomials by their exponent vectors $e \in \Z_{\geq 0}^{n}$ and write $\var{x}^e \defeq x_1^{e_1} \dotsc x_n^{e_n}$; we denote the corresponding coefficient by $a_e$, and then write any polynomial as $f(\var{x}) = \sum_{e \in \Z_{\geq 0}^n} a_e \var{x}^e$ (only finitely many terms will be nonzero). If $p \in \C[(a_e)_{e \in \Z_{\geq 0}^n}]$ is a test polynomial and $f = \sum_{e} \alpha_e \var{x}^e$ is an input polynomial, we write $p(f)$ for the evaluation of $p$ in which each test variable $a_e$ is set to the corresponding coefficient $\alpha_e \in \F$ of $f$.

\begin{defn} \label{def:poly}
A property $\Pi$ of input polynomials is \definedWord{defined by (test) polynomials} if there is a set of test polynomials $p_1, \dotsc, p_k$ such that $f(\var{x})$ has property $\Pi$ if and only if $p_1(f) = p_2(f) = \dotsc = p_k(f) = 0$.
\end{defn}

\begin{rmk}
Readers familiar with algebraic geometry will note that a property defined by test polynomials is exactly the same thing as an algebraic subset of the vector space $\Poly^d(x_1, \dotsc, x_n)$. This is an algebro-geometric viewpoint on complexity. We discuss this further in Section~\ref{sec:necessary}. For now we note that such algebro-geometric notions of complexity have been used before: border rank for matrix multiplication and ``infinitesimal approximation'' in GCT are both algebro-geometric notions of complexity in this sense.
\end{rmk}

\begin{rmk}
By Hilbert's Basis Theorem, any property defined by polynomials can be defined by finitely many polynomials.
\end{rmk}

\subsection{Linear-invariant properties defined by polynomials}
Kayal \cite[Sec.~5.2]{kayal} observes that several lower bounds use linear-invariant properties at their core, and in fact this observation was the starting point for this paper. In this paper we extend this observation in two directions simultaneously: (1) we observe that most arithmetic circuit lower bounds use (linear-)invariant properties \emph{defined by polynomials} (Definition~\ref{def:poly}), allowing us to make the connection with representation theory and GCT, and (2) we extend the observation to most arithmetic circuit lower bounds.

\begin{defn}
A property $\Pi$ of (input) polynomials is \definedWord{linear-invariant} if for every polynomial $f(x_1, \dotsc, x_n)$ and every invertible linear change of variables $A \in \GL_n(\F)$
\[
f(\var{x}) \text{ has property $\Pi$} \iff f(A\var{x}) \text{ has property $\Pi$}
\].
\end{defn}

\begin{example}
The property of being a perfect square is linear-invariant: $f(\var{x}) = g(\var{x})^2$ if and only if $f(A\var{x}) = g(A\var{x})^2$ for any invertible linear change of variables $A$. As explained in the previous section, in the case of $f(x,y)$ homogeneous of degree 2, this property is defined by the vanishing of the test polynomial $b^2 - 4ac$.
\end{example}

\begin{example} \label{ex:pder}
The dimension of the space of all partial derivatives of a homogeneous polynomial $f$ is a linear-invariant property. The $k$-th order partial derivatives of $f$ are linearly independent from its $\ell$-th order partial derivatives for $k \neq \ell$, so we may prove this for each $k$ separately. Consider the partial derivative $\left(\frac{\partial f}{\partial x_{i_1} \dotsb \partial x_{i_k}}\right)(\var{x})$. When we transform the variables $\var{x}$ by $A$, we change both the variables with respect to which the derivatives are being taken, and we change the variables at which the partial derivative is being evaluated. The fact that the former kind of transformation does not change the dimension of the space of partial derivatives follows from the usual ``directional derivative'' formula from multilinear calculus. The latter kind of transformation also does not change the dimension of a space of polynomials, for $\sum_{i=1}^{d} \alpha_i g_i(\var{x}) = 0$ if and only if $\sum_{i=1}^d \alpha_i g_i(A\var{x}) = 0$. We will see below that this property is also defined by polynomials.
\end{example}

The notion of a linear-invariant property defined by polynomials is embodied in the following definition. To make the definition clear we first introduce one more bit of notation. Each linear change of input variables $B \in \GL_n(\F)$ defines a linear map $\Coeff_{B}$ from $\Poly^d(x_1, \dotsc, x_n)$ to itself: $B$ sends $f(\var{x}) = \sum_e a_e \var{x}^e$ to $f(B\var{x}) = \sum_e a'_e \var{x}^e$. In other words, $\Coeff_{B}$ is the linear map taking the coefficient vector $(a_e)_{e \in \Z_{\geq 0}^n}$ to the new coefficient vector $\Coeff_{B}((a_e)_e) = (a'_e)_e$. It is a standard fact---easily verified---that $\Coeff_{B}$ is linear\footnote{Linear in the coefficients $a_e$. It will have degree $d$ in the coordinates of $B$, but that is not relevant here.}. Thus $B$ induces a linear map $\Coeff_{B}$ on the coefficients of input polynomials, which are in turn the variables of test polynomials. Then $\Coeff_{B}$ induces a linear map on test polynomials, taking $p((a_e)_e)$ to $p(\Coeff_{B}((a_e)_e))$.

\begin{defn}
A \definedWord{test $\GL_n(\F)$-module}\footnote{See Appendix~\ref{app:terminology} for a discussion of the terminology.} is a finite-dimensional vector space $T$ of test polynomials, say with basis $\{p_1, \dotsc, p_k\}$, such that for each $1 \leq i \leq k$ and each $B \in \GL_n(\F)$, $p_i(\Coeff_{B}((a_e)_e))$ lies in $T$.
\end{defn}

We say a test module $T$ vanishes on an input polynomial $f$ if every test polynomial $p \in T$ vanishes at $f$. The set of input polynomials at which a given test module vanishes is a linear-invariant set, which we can think of as a linear-invariant property:

\begin{fact} \label{fact:propertiesModules}
There is a many-to-one correspondence between test $\GL_n(\F)$-modules and linear-invariant properties defined by polynomials. 
\end{fact}

That is, each linear-invariant property defined by polynomials is defined by some test $\GL_n(\F)$-module, and each test $\GL_n(\F)$-module defines a linear-invariant property\footnote{If $\F$ is algebraically closed, then Hilbert's Nullstellensatz implies that two test modules $T_1, T_2$ define the same invariant property if and only if the ideals of test polynomials generated by $T_1$ and $T_2$ have the same radical. Recall that the radical of an ideal $I$ is the ideal $\sqrt{I} \defeq \{ f : f^k \in I \text{ for some } k \geq 1\}$.}. The proof involves only basic observations regarding group actions and algebraic sets (see Appendix~\ref{app:invariant}).

\begin{example}
The vector space spanned by the test polynomial $b^2-4ac$ is a one-dimensional test $\GL_2(\F)$-module. For let $f(x,y) = a x^2 + b x y + c y^2$ and $A = \left(\begin{array}{cc} \alpha & \beta \\ \gamma & \delta \end{array} \right)$, and write $f(A\var{x}) = a (\alpha x + \beta y)^2 + b (\alpha x + \beta y)(\gamma x + \delta y) + c (\gamma x + \delta y)^2 = a' x^2 + b' x y + c' y^2$. Let $p(a,b,c) = b^2 - 4 ac$; then $p(\Coeff_{A}(a,b,c)) = p(a',b',c') = b'^2 - 4 a' c'$. A simple but tedious calculation then reveals that $b'^2 - 4a'c' = \det(A)^2\left(b^2 - 4ac\right)$, and hence that $p(\Coeff_{A}(a,b,c))$ is a scalar multiple of $p(a,b,c)$. 
\end{example}

\subsection{Separating modules and a first example}
\begin{defn} \label{def:sep}
A \definedWord{separating module}\footnote{Separating modules are nearly equivalent to the ``HWV obstructions'' of B\"{u}rgisser and Ikenmeyer \cite{BI}. For a discussion of the exact relationship and choice of terminology see Appendix~\ref{app:terminology}.} for the lower bound $\mathcal{C}_{hard} \not\subseteq \mathcal{C}_{easy}$ is a test module $T$ such that $T$ vanishes on every function in $\mathcal{C}_{easy}$, but does not vanish at some function $f_{hard} \in \mathcal{C}_{hard}$.
\end{defn}

The main thesis of this paper is that most arithmetic circuit lower bounds already use separating modules. We now demonstrate this with an example, by showing that Theorem~0 of Nisan and Wigderson \cite{nisanWigderson} uses a separating module. We first recall their definitions and result. In the next section we argue that the existence of a separating module was in some sense necessary.

An arithmetic circuit is \definedWord{homogeneous} if every gate in the circuit computes a homogeneous polynomial. The $d$-th elementary symmetric function in $n$ variables is the sum of all multilinear monomials of degree $d$ and is denoted $e_{d,n}$.

\begin{thm}[{\cite[Thm.~0]{nisanWigderson}}] \label{thm:NW} Over a field of characteristic zero, any homogeneous depth 3 arithmetic circuit computing $e_{2d,n}$ has size $\Omega\left( \left(\frac{n}{4d}\right)^{d}\right)$.
\end{thm}

When $d = c n$ for any $0 < c \leq 1/4$, this lower bound is exponential in $n$.

\begin{proof}[Proof outline]
The key property they consider is the dimension of the space of all partial derivatives (of all orders) of a function. We denote this space $\partial(f)$. First, they show that $\dim \partial(C) \leq s 2^{d}$ for any homogeneous depth 3 arithmetic circuit $C$ of size $s$ computing a polynomial of degree $d$. Next, they show that $\dim \partial(e_{2d,n}) \geq \binom{n}{d}$. Combining these inequalities, one gets $s2^{2d} \geq \binom{n}{d} \geq \left(\frac{n}{d}\right)^{d}$. 
\end{proof}

\begin{prop} \label{prop:example0} There is a separating module for the lower bound of Theorem~\ref{thm:NW}. \end{prop}

\begin{proof}
Let $\Pi(r)$ denote the property ``$\dim \partial(f) \leq r$.'' We argued in Example~\ref{ex:pder} that $\dim \partial(f)$ is a linear-invariant property for homogeneous $f$. We now show that this property is defined by a test $\GL_{n}(\F)$-module, and hence that the above proof yields a separating module.

Let $f(\var{x}) = \sum_e a_e \var{x}^e$ be a homogeneous polynomial of degree $d$ (the only nonzero terms in the sum are those for which $\sum_{i} e_{i} = d$) and consider the following matrix $M_f$. The columns of $M_f$ are indexed by the monomials of degree $\leq d$, and the rows of $M_f$ are indexed by the partial derivative operators (these are in bijective correspondence with monomials, but we refer to them this way to keep track of which is which). The entry in the $\partial^k / \partial x_{i_1} \dotsb \partial x_{i_k}$ row and the $\var{x}^e$ column is the coefficient of $\var{x}^e$ in $\partial^k f / \partial x_{i_1} \dotsb \partial x_{i_k}$. Note that this coefficient is some linear combination of the coefficients $a_e$ of $f$. 

Then the dimension of $\partial(f)$ is the same as the (row) rank of $M_f$. It is a standard fact from linear algebra that $M_f$ has rank $\leq r$ if and only if all the $(r+1) \times (r+1)$ minors of $M_f$ vanish. Each such minor is a degree $r+1$ polynomial of the entries of $M_f$, which are themselves linear combinations of the coefficients $a_e$ of $f$. Hence each such minor is a test polynomial of degree $r+1$. Let $T(r)$ denote the linear span of these minors. We have just shown that (the vanishing of the test polynomials in) $T(r)$ defines the property $\Pi(r)$.

In particular, $\Pi(r)$ is a linear-invariant property defined by polynomials. By Fact~\ref{fact:propertiesModules} $\Pi(r)$ is defined by some test module, which is thus a separating module. However, we can argue further that $T(r)$ itself is a test $\GL_n(\F)$-module, and hence a separating module for the lower bound of Theorem~\ref{thm:NW}.

In Example~\ref{ex:pder} we essentially showed that $M_{f(A\var{x})}$ is related to $M_{f(\var{x})}$ by left and right multiplication by some matrices related to $A$ (in a similar way to how $\Coeff_{A}$ is related to $A$). It is a standard fact about minors that the $(r+1) \times (r+1)$ minors of $BM_f C$ are linear combinations of the $(r+1) \times (r+1)$ minors of $M_f$. Hence for any test polynomial $p \in T(r)$, $p \circ \Coeff_{A}$ is also in $T(r)$. Thus $T(r)$ is a separating module for Theorem~\ref{thm:NW}.
\end{proof}

As with everything in complexity, in fact what we have is a family of separating modules. Namely, if we consider $e_{2d,n}$ with $d=n/8$, then $T(2^{n/8})$ vanishes at every polynomial computed by a depth 3 homogeneous circuit of degree $n/4$ and size at most $2^{n/8}$, but does not vanish at $e_{n/4,n}$. 

\subsection{Generalizations} \label{sec:general}
For other lower bounds it is useful to generalize some of the above notions.

First, we can allow input objects other than input polynomials. For example, in the context of matrix rigidity it will be useful to consider input matrices. Regardless of the input objects, we still speak of test polynomials. In the case of input matrices, test polynomials are then polynomials whose variables are the coordinates $a_{ij}$ of the input matrices. In the context of Boolean functions, we often first represent a function by its unique multilinear polynomial, and then work in the context of input polynomials. But one could imagine a more direct representation in terms of something like ``input circuits.'' In the context of the degree bound \cite{strassenDegree, baurStrassen} and the connected components sorting lower bound \cite{benOr}, the input objects are (semi-)algebraic sets, given by their defining polynomial (in)equalities. The variables for the test polynomials are then the coefficients of the equations defining the algebraic sets.

Second, we can allow other types of invariance besides linear invariance. For example, we can hardly imagine a complexity measure or lower bound proof that depends on the order or names of the variables. Hence all properties used in complexity can be expected to be \definedWord{permutation-invariant}: $f(x_1, \dotsc, x_n)$ has the property if and only if $f(x_{\pi(1)}, \dotsc, x_{\pi(n)})$ has the property, for any permutation $\pi$. We then speak of test $S_n$-modules, and the analog of Fact~\ref{fact:propertiesModules} holds (see Fact~\ref{fact:propertiesModulesGeneral}). Note that $S_n$-modules are still defined as vector spaces; the use of vector spaces in the definition of test module was not specific to $\GL_n$. We will use permutation-invariance in the contexts of matrix rigidity and multilinear formulas and circuits, as these concepts are not linear-invariant but they are permutation-invariant.

Another type of invariance that often arises is affine invariance. Here we generalize from linear transformation $\var{x} \mapsto A\var{x}$ to affine transformations $\var{x} \mapsto A\var{x} + \var{b}$, with $A \in \GL_n(\F)$ and $\var{b} \in \F^{n}$. The group of all such transformations is the affine general linear group $\AGL_n(\F)$. We then speak of \definedWord{affine-invariant} properties and test $\AGL_n(\F)$-modules. Again, the analog of Fact~\ref{fact:propertiesModules} holds (see Fact~\ref{fact:propertiesModulesGeneral}).

When the invariance is understood from context, we may simply refer to test modules and separating modules without reference to a particular group.

\section{On the necessity and utility of separating modules and border complexity} \label{sec:necessary}
In Section~\ref{sec:invariant} we argue that the use of invariant properties is essentially necessary. In Section~\ref{sec:border} we discuss situations where furthermore the use of separating modules is essentially necessary. Although not all complexity classes are defined by the vanishing of test polynomials, in Section~\ref{sec:constructible} we argue that all nonuniform complexity classes, including Boolean ones, are ``constructible'' by test polynomials (see Definition~\ref{def:constructible}). Finally, in Appendix~\ref{app:necessary} we give a heuristic argument as to why separating modules are likely to be the easiest way to prove lower bounds against constructible complexity classes, and shed light on a complexity class even when their use is not strictly necessary. Hence separating modules should be a first approach to try. We defer this final argument to an appendix only because it is heuristic, somewhat technical, and possibly contentious, and we do not wish to distract from the main points of the paper. However, one argument for this which we can already state is that \emph{most arithmetic circuit lower bounds already use separating modules}, as shown in this paper.

Throughout this section and Appendix~\ref{app:necessary}, we only discuss nonuniform lower bounds. If $\mathcal{C}$ is a nonuniform complexity class, then $\mathcal{C}_n$ denotes the functions in $\mathcal{C}$ with $n$ inputs. By a ``property'' in general, we mean a set of input polynomials, or more generally input objects.

\subsection{Invariant properties are necessary} \label{sec:invariant}
First we show that if $\mathcal{C}_n$ is invariant under some group $G$---such as $\GL_n$, $S_n$, etc.---then any property used to prove a lower bound against $\mathcal{C}_n$ can be transformed into a $G$-invariant property that proves the same lower bound. Then we argue that essentially all ``naturally occurring'' complexity classes and complexity measures are permutation-invariant, and many are linear- or affine-invariant. 

\textbf{If any property can be used against a $G$-invariant class, a $G$-invariant property can.} Suppose property $\Pi$ is used to prove a lower bound\footnote{For readers familiar with Natural Proofs \cite{razborovRudich}, note that we are using the complementary notion of ``useful property'' here. They use properties $\Pi$ that are disjoint from $\mathcal{C}_n$, whereas we use properties $\Pi$ that completely contain $\mathcal{C}_n$. By taking the complements of sets, the two viewpoints are equivalent. We chose our viewpoint because it has nicer algebro-geometric properties, as in Appendix~\ref{app:necessary}.} against $\mathcal{C}_n$ by showing that $\mathcal{C}_n \subseteq \Pi$ and $f_{hard,n} \notin \Pi$. Let $\Pi^{G}$ denote the unique maximum $G$-invariant subset contained in $\Pi$; this exists by Zorn's Lemma, as an arbitrary union of $G$-invariant subsets is $G$-invariant. As $\mathcal{C}_n$ is $G$-invariant, by the definition of $\Pi^{G}$ we have $\mathcal{C}_n \subseteq \Pi^{G}$. The $G$-invariant property $\Pi^{G}$ then proves the same lower bound as $\Pi$, as $f_{hard,n} \notin \Pi \supseteq \Pi^{G} \supseteq \mathcal{C}_n$.

\textbf{Essentially all complexity classes are permutation-invariant.} All complexity measures and complexity classes we are aware of are permutation-invariant: they do not depend on the names or order of the variables. Indeed, we imagine that any complexity class or measure that was \emph{not} permutation-invariant would be quite perverse, as the complexity of computing a function should really not depend on whether its variables are called $x_1, \dotsc, x_n$ or $a,b,c,\dotsc$, or $x_n, \dotsc, x_1$. Thus we can expect that any lower bound uses a permutation-invariant property, at the very least.

Many complexity classes, particularly algebraic ones, are furthermore linear- or affine-invariant. For example, arithmetic circuit size does not change by more than an additive difference of $n$ after a linear or affine transformation\footnote{In the model where addition gates can compute linear or affine combinations; in the weaker model where addition gates are just addition gates, the size still does not change by more than $O(n^2)$.}. Additionally, circuit depth increases by at most 1; for circuits whose bottom gates are linear combination gates, the depth need not increase at all. For example, $\cc{AC}^0[2]$ is $\AGL_n(\F_2)$-invariant (though we note that $\cc{AC}^0$ is not $\GL_n(\F_2)$-invariant, as $\F_2$-linear transformations are as powerful as parity). This is in line with Kayal's initial observation \cite[Sec.~5.2]{kayal} that several known lower bounds use affine-invariant properties, and with our observations in this paper.

Hence, for all naturally occurring nonuniform complexity classes, if any property can be used to prove a lower bound, a permutation-invariant property can be used.

\subsection{Test polynomials and border complexity} \label{sec:border}
A complexity class $\mathcal{C}_n$ is typically not defined by the vanishing of some test polynomials. Hence when we prove a lower bound against $\mathcal{C}_n$ using test polynomials, we in fact prove a lower bound against the slightly larger class which we denote $\overline{\mathcal{C}_n}$ and refer to as ``border-$\mathcal{C}_n$,'' in line with normal usage in other contexts (the overline is for Zariski-closure; see Definition~\ref{def:constructible}). Standard results in algebraic geometry (\eg, \cite[Thm.~2.33]{mumford}, \cite[\S 20.6]{BCS}) imply that $\overline{\mathcal{C}_n}$ consists of all functions $f$ which can be written as a limit\footnote{Over $\C$ the notion of limit is defined in the usual manner. Over, say, $\overline{\F}_p$, we say a function $f$ is a limit of points in $\mathcal{C}_n$ if there is a one-dimensional family of functions $f_{t}$ such that $f_{t}$ is well defined and in $\mathcal{C}_n$ for all but finitely many values of $t \in \overline{\F}_p$, and $f_0 = f$. There is one additional technical condition here, but we omit it since it does not affect our discussion.} of functions in $\mathcal{C}_n$.

In the next section we show that $\mathcal{C}$ is not too far from border-$\mathcal{C}$. In Appendix~\ref{app:necessary} we argue that proving lower bounds against border-$\mathcal{C}$ is likely to be the easiest way to prove lower bounds against $\mathcal{C}$, despite being a formally stronger statement. Here we present examples where there is known to be little or no difference, and begin arguing for the utility of border complexity.

\begin{example}[Matrix multiplication]
In the context of matrix multiplication the typical complexity measure is \emph{tensor rank}, which is essentially the number of non-scalar multiplications needed to multiply two matrices. Tensor rank is known to agree with the total number of arithmetic operations up to a constant factor. The corresponding border complexity measure is called \emph{border rank}, or sometimes ``approximative complexity,'' first introduced by Bini, Capovani, Lotti, and Romani \cite{BCRL}. In general, border rank can be smaller than tensor rank. However, Bini \cite{bini} showed that the exponent of matrix multiplication calculated with tensor rank---the smallest $\omega$ such that $n \times n$ matrix multiplication has tensor rank $O(n^{\omega})$---is the same as the exponent calculated with border rank. Thus, although border rank and tensor rank are not equal, they give the same asymptotic answer for matrix multiplication.

Furthermore, the use of border rank has greatly increased our understanding of both upper and lower bounds for matrix multiplication. One of the main tools for finding efficient algorithms for matrix multiplication is Sch\"{o}nhage's asymptotic sum inequality \cite{schonhage}, which shows that an upper bound on border rank implies an upper bound on tensor rank. Conversely, most lower bounds on matrix multiplication seem to have a border rank lower bound at their heart. For example, Landsberg \cite[\S 6]{landsbergMMSurvey} showed that Bl\"{a}ser's tensor rank lower bound \cite{blaser}---the then best known bound---implicitly uses the same key lemma that Strassen used \cite{strassenBorderRk} to give a border rank lower bound. The currently best known lower bound on tensor rank \cite{landsberg2,massarentiRaviolo} also uses techniques from the best known lower bound on border rank \cite{landsbergOttaviani}. 
\end{example}

\begin{example}[Permanent versus determinant] 
In the context of permanent versus determinant, the typical complexity measure is \emph{determinantal complexity}: the size of the smallest matrix $M(\var{x})$ with linear combinations of the variables $\var{x}$ for entries such that $\det(M(\var{x})) = \perm(\var{x})$. Mulmuley and Sohoni \cite{gct1} use the analogous notion of border determinantal complexity, which they refer to as ``infinitesimal approximative'' complexity. Independently, B\"{u}rgisser, Landsberg, Manivel, and Weyman \cite[Prop.~9.4.3]{BLMW} and the author \cite[Prop.~3.5.4]{grochowPhD} show that under certain fairly general circumstances the border determinantal complexity only differs from the determinantal complexity by a polynomial, and state a conjecture which would imply this is always the case. Thus border complexity here is not as far from standard complexity as it may at first seem.

In contrast, Mulmuley and Sohoni \cite[\S 4.2]{gct1} give an example of a function which has border determinantal complexity $poly(n)$ but which may have super-polynomial determinantal complexity. Such functions exhibit a difference in the difficulties of resolving the complexity of matrix multiplication and resolving the permanent versus determinant problem. Nonetheless, they conjecture \cite[Conj.~4.3]{gct1} that no $\cc{VNP}$-hard function has polynomial border determinantal complexity. One might also guess that for quasi-polynomial complexity there is no difference, that is, that the following question has a positive answer:

\begin{open} \label{open:border}
Does polynomial, or more generally quasi-polynomial, border determinantal complexity imply quasi-polynomial determinantal complexity? Equivalently, is $\overline{\cc{VP}_{ws}} \subseteq \cc{VQP}$ or more generally $\overline{\cc{VQP}} = \cc{VQP}$?
\end{open}
\end{example}

Either way, as all of our current techniques give bounds on border complexity, Question~\ref{open:border} is an archetype of a fundamental question of the difference between the way complexity classes are usually defined and the methods we use for proving lower bounds against them.

Because of the above results and the prevalent use of test polynomials in known lower bounds, as well as the arguments in Appendix~\ref{app:necessary}, we submit that border complexity in general---not only in the context of matrix multiplication---is a natural and useful measure of complexity from the perspective of lower bounds (and, in the context of matrix multiplication, upper bounds as well!).

\subsection{Nonuniform complexity classes are constructible by test polynomials}
\label{sec:constructible}
Over any field, if $\mathcal{C}_n$ is defined by test polynomials, say $\mathcal{C}_n = \{ f | t_1(f) = t_2(f) = \dotsb = t_k(f) = 0\} \subseteq \Poly^{d(n)}(x_{1}, \dotsc, x_n)$, then $f_{hard,n} \notin \mathcal{C}_n$ if and only if there is some $1 \leq i \leq k$ such that $t_i(f_{hard}) \neq 0$. For such classes, the use of test polynomials is necessary and sufficient to prove a lower bound. However, most complexity classes are not defined by test polynomials in this manner. We will argue here that all naturally occurring complexity classes are nonetheless ``constructible'' by test polynomials (definition below). In Appendix~\ref{app:necessary} we argue that test polynomials---and hence, via Section~\ref{sec:invariant} and Fact~\ref{fact:propertiesModules}, separating modules---are nonetheless incredibly useful for understanding such constructible (invariant) complexity classes.
 
\begin{defn}[Zariski, \ie\ algebro-geometric, topology] \label{def:constructible}
A set defined by the vanishing of test polynomials is called \emph{(Zariski-)closed}. A set is \definedWord{constructible} if it can be constructed from closed sets by taking complements, unions, and intersections.
\end{defn}

The \emph{closure} of a set $S$ is the smallest closed set containing $S$, and is denoted $\overline{S}$. If $S$ is a Zariski-constructible set over $\C$, then its Zariski-closure coincides with its closure in the usual complex topology (see, \eg, \cite[Thm.~2.33]{mumford}). Note that the closure $\overline{S}$ is the set of all points which cannot be separated from $S$ by test polynomials. 

The main insight of this section is a corollary to Chevalley's constructibility theorem. To state this theorem, we need one more concept. A map $\varphi\colon A \to B$ between constructible sets is called \emph{algebraic} if its graph $\{(a, \varphi(a)) | a \in A \}$ is a closed subset of $A \times B$. Equivalently, let $x_1, \dotsc, x_n$ be coordinates on $B$, not necessarily independent; then $\varphi$ is algebraic if and only if for each $i$, $x_i(\varphi(a))$ can be expressed as a polynomial in the coordinates of $a \in A$. 

Chevalley's Theorem is most concisely stated for Noetherian rings, but we will not need their definition here. For our purposes it suffices that this includes $\Z$, $\Z/n\Z$, rings of algebraic integers, all fields, polynomial rings, and quotients of polynomial rings.

\begin{thm}[{Chevalley's Theorem\footnote{The original version of this theorem over algebraically closed fields is from Chevalley and Cartan \cite{chevalleyCartan}. The general version, which is in fact more general than stated here, can be found as {\cite[Thm.~1.8.4]{EGA4}}. See Eisenbud {\cite[Cor.~14.7]{eisenbud}} for a purely ring-theoretic treatment of what is essentially the general case, or Matsumura {\cite[Ch.~1, \S 6]{matsumura}}.}}] Over any Noetherian ring the image of any algebraic map is constructible. 
\end{thm}

We are not aware of any nonuniform complexity classes---algebraic or otherwise---that do not belong to one of the classes described in the following corollary:

\begin{cor} \label{cor:constructible}
Let $\mathcal{C}$ be a nonuniform complexity class; then $\mathcal{C}_n$ is (Zariski-)constructible if any of the following hold:
\begin{enumerate}
\item $|\mathcal{C}_n|$ is finite; or

\item \label{cond:map} $\mathcal{C}$ is closed under simple (resp. linear, resp. affine) projections, and contains a problem that is complete under simple (resp. linear, resp. affine) projections; or

\item \label{cond:cktshape} $\mathcal{C}_n$ is defined by a class of circuits that are restricted to have one of finitely many (a number which may grow with $n$) shapes. Here by the ``shape'' of a circuit, we mean the underlying directed acyclic graph together with operators labeling the internal nodes; or

\item More generally, $\mathcal{C}_n$ is first-order definable in the language of rings over a Noetherian ring, or in the language of ordered rings over an ordered Noetherian ring.
\end{enumerate}
\end{cor}

A ``simple projection'' here means any map that sends each variable $x_i$ to a constant $\alpha$ or to a constant multiple of a variable $\alpha y_j$. A linear projection sends each $x_i$ to a linear combination of variables $\sum_j \alpha_{ij} y_j$, and an affine projection additionally allows an additive constant: $x_i \mapsto \alpha_i + \sum_j \alpha_{ij} y_j$. 

Condition (\ref{cond:cktshape}) includes circuit classes defined in terms of fan-in, size, depth, or connectivity properties like skew or weakly-skew. 

\begin{proof}
(1) Any finite set is defined by the vanishing of test polynomials, \ie\ it is closed, hence constructible.

(2) The set of simple (resp. linear, resp. affine) projections is closed, as we show below; denote this set by $R$, for ``reductions.'' If $f_n$ is a complete function, and $F$ is the space of input functions (objects, etc.), then define a map $\varphi\colon R \to F$ by $\varphi(r) = r(f_n)$. From the definition of projection, it is easily seen that $\varphi$ is algebraic. Then $\mathcal{C}_n$ is the image of $\varphi$, hence is constructible by Chevalley's Theorem.

The set of linear (resp. affine) projections from functions on $n$ variables to functions on $m$ variables is just the set of $m \times n$ (resp. $(m+1) \times n$) matrices, so is closed. The set of simple projections is the subset of affine projections defined by the property that each column of the $(m+1) \times n$ matrix has at most one nonzero entry. The latter condition is equivalent to the condition that the product of any two entries from a given column vanishes, hence the set of simple projections is closed.

(3) For each circuit shape $G$, the set of circuits of that shape is $\F^{N}$ where $N$ is the number of edges whose endpoints are linear combination gates. Let $\text{Ckt}_{G}$ denote this space, and let $\varphi_{G}\colon \text{Ckt}_{G} \to \Poly^d(x_1, \dotsc, x_n)$ be the map which takes each circuit of shape $G$ to the function it computes. It is easily seen that $\varphi_{G}$ is algebraic, so its image is constructible by Chevalley's Theorem. Then $\mathcal{C}_n$ is the union over finitely many shapes $G$ of $\im(\varphi_{G})$. As a union of constructible sets is constructible, so is $\mathcal{C}_n$.

(4) A first-order definable set is defined by some first-order formula. For quantifier-free formulas, this is exactly a set defined by a logical combination of equalities and inequalities, namely a constructible set. The only tricky part is then to handle quantifiers. By replacing a universal quantifier $\forall x$ by $\neg \exists x \neg$ and noting that the complement of a constructible set is constructible, we need only handle existential quantifiers. If $\varphi(\var{x})$ is a first-order formula without quantifiers, let $\mathcal{C}'$ denote the set of those $\var{x}$ that satisfy $\varphi(\var{x})$. Then the set defined by $\exists x_0 \varphi(\var{x})$ is equal to the image of $\mathcal{C}'$ under the projection which sends $(x_0, x_1, \dotsc, x_n) \mapsto (x_1, \dotsc, x_n)$. By Chevalley's Theorem, the image of this projection is constructible.
\end{proof}

Note that if a circuit class is defined as the image of some map---as nearly all of them are, as in conditions (\ref{cond:map}) and (\ref{cond:cktshape})---finding its representation as a union of differences of closed sets may be difficult, even uncomputable. However, over finite fields this is a finite problem, hence computable, and over algebraically closed fields or real closed fields quantifier elimination algorithms such as Tarski's \cite{tarski} make this process effective.

\begin{rmk}
The (Zariski-)closure of classes satisfying condition (\ref{cond:map}) of Corollary~\ref{cor:constructible} for linear or affine projections are \emph{orbit closures} for $\GL_n$, respectively $\AGL_n$. Much of the current research in GCT studies the orbit closures associated to the permanent, determinant, and matrix multiplication. Considering their structure as orbit closures rather than just $G$-invariant sets facilitates their study greatly, much as the existence of complete problems facilitates the study of a complexity class. In this paper we show that by extending our viewpoint to all $G$-invariant complexity classes and not just orbit closures, GCT becomes much more general and far-reaching.
\end{rmk}

\section{Discussion, relation to the GCT program, and future directions} \label{sec:discussion}
In this paper, we show that most arithmetic circuit lower bounds and implications between lower bounds fit naturally into the representation-theoretic framework suggested by geometric complexity theory, specifically in the form of separating modules. In this section we discuss further implications of this connection, as well as which lower bounds seem to not fit into this framework (ones which are essentially uniform), the status of lower bounds in positive characteristic, and the relation between this work and the larger GCT program.

Eric Allender observed that all the lower bounds mentioned here use Razborov--Rudich-natural \cite{razborovRudich} properties\footnote{\label{fn:RR} The Boolean properties satisfy the Razborov--Rudich conditions, and although there is no known algebraic analog of the Razborov--Rudich barrier, the algebraic properties mentioned in the previous sections seem like they ought to fulfill the requirements of such an analog, were it to exist.}, and asked whether this was just a coincidence. In light of the generality of separating modules (Section~\ref{sec:necessary}), we believe that it is indeed a coincidence, and has more to do with the fact that most known results use such properties than it has to do with any inherent limitations of the representation-theoretic viewpoint. Indeed, there is evidence that the GCT approach over $\C$ philosophically (see Footnote~\ref{fn:RR}) avoids the Razborov--Rudich barrier (the author's thesis \cite[Sec.~3.4.3]{grochowPhD} contains an overview of such evidence).

\subsection{Relation to Geometric Complexity Theory} \label{sec:gct}
To state how the separating modules used in this paper differ from the geometric obstructions defined in Mulmuley and Sohoni \cite{gct2}, and to discuss possible further interactions between previous results and geometric complexity theory, we first recall two standard definition from representation theory, as applied to test modules. Test $G$-modules for any group $G$---such as $\GL_n(\F)$, $S_n$, etc.---are, in particular, representations of $G$; indeed, the term ``module'' is often used interchangeably with ``representation'' (see Appendix~\ref{app:terminology} for more on the terminology). When we consider a test $G$-module as just a representation of $G$ (equivalently, as a $G$-module), we forget that it consists of test polynomials, and only remember that it is a vector space and how the elements of $G$ move vectors around within this vector space.

\begin{defn}
A (test) $G$-module $T$ is \definedWord{irreducible} if there is no nonzero proper subspace of $T$ that is also a (test) $G$-module.
\end{defn}

A classical theorem (see, \eg, \cite{fultonHarris}) says that over an algebraically closed field of characteristic zero, every $\GL_n$- or $S_n$-module is a direct sum (as representations, that is, as vector spaces) of irreducible submodules. In particular, this implies that if there is a separating module for a lower bound over $\C$, there is an irreducible separating module. We could have included irreducibility in the definition of test module for this reason, but chose not to in order to keep the definition simple and to avoid complications over other fields, especially finite fields. The property of splitting into a direct sum of irreducible submodules is known as ``complete reducibility.'' It is known to fail in general for $\AGL_n$-modules, even over $\C$, and for $\GL_n$- and $S_n$-modules in positive characteristic.

\begin{defn} \label{def:equiv}
Two (test) $\GL_n$-modules $T_1, T_2$ are \definedWord{equivalent} (as representations) if there is a bijective linear map $L\colon T_1 \to T_2$ such that for all $A \in \GL_n$ and all test polynomials $p \in T_1$, $L(p \circ \Coeff_{A}) = L(p) \circ \Coeff_{A}$.
\end{defn}

This definition is purely representation-theoretic, in that it ignores the ``test polynomial'' structure of the test modules, and treats them only as representations. Because this notion of equivalence forgets the underlying polynomials of the test modules, it is possible---and is likely to be the generic situation---for two equivalent test $\GL_n$-modules to define distinct linear-invariant properties. Nonetheless, the term ``equivalent'' is standard in representation theory, so we use it here.

To discuss the geometric obstructions of GCT, we work over $\C$. By complete reducibility, the space of all test polynomials can be written as a direct sum of test $\GL_n(\C)$-modules. If we group these modules by their equivalence classes, we may write the space of all test polynomials as the direct sum $\bigoplus_{\lambda} \bigoplus_{i=1}^{m_{\lambda}} T_{\lambda,i}$ where the $\lambda$\,s index the irreducible equivalence classes. (An equivalence class is called irreducible if any representation in this class is irreducible.) It turns out that each equivalence class $\lambda$ can only occur amongst a specific degree $d(\lambda)$ of test polynomials, and since the space of test polynomials of any fixed degree is finite-dimensional, each $m_{\lambda}$ is finite. Moreover, the numbers $m_{\lambda}$ are independent of the choice of direct sum. We refer to $m_{\lambda}$ as the \definedWord{multiplicity} of the equivalence class $\lambda$ in the space of test polynomials. 

If $\mathcal{C}$ is a linear-invariant complexity class, and we consider the space of all test polynomials that vanish everywhere on $\mathcal{C}$ (hence on $\overline{\mathcal{C}}$, see Section~\ref{sec:border}), we may write this space as $\bigoplus_{\lambda} \bigoplus_{i=1}^{m_{\lambda}(\overline{\mathcal{C}})} T_{\lambda,i}$. Note that $m_{\lambda}(\overline{\mathcal{C}}) \leq m_{\lambda}$. 

\begin{defn}[Mulmuley and Sohoni \cite{gct2}\footnote{Although only geometric obstructions were explicitly defined in \cite{gct2}, multiplicity obstructions were essentially defined there: see the sentence just before \cite[Def.~1.2]{gct2}.}] A \definedWord{multiplicity obstruction} for the lower bound $\mathcal{C}_{hard} \not\subseteq \mathcal{C}_{easy}$ is an irreducible equivalence class $\lambda$ such that $m_{\lambda}(\overline{\mathcal{C}_{easy}}) > m_{\lambda}(\overline{\mathcal{C}_{hard}})$. A \definedWord{occurrence obstruction} or \definedWord{geometric obstruction} for $\mathcal{C}_{hard} \not\subseteq \mathcal{C}_{easy}$ is a multiplicity obstruction which further has $m_{\lambda}(\overline{\mathcal{C}_{easy}}) = m_{\lambda}$, that is, every test module equivalent to $\lambda$ vanishes on $\mathcal{C}_{easy}$.
\end{defn}

The existence of a multiplicity obstruction $\lambda$ implies the existence of a separating module, as then there must be some test $\GL_{n}(\C)$-module of type $\lambda$ that vanishes on $\mathcal{C}_{easy}$ but not on $\mathcal{C}_{hard}$. These are referred to as ``obstructions'' because they obstruct the inclusion $\mathcal{C}_{hard} \subseteq \mathcal{C}_{easy}$, much as a $K_5$-minor obstructs a planar embedding of a graph. 

One advantage of considering multiplicities rather than test modules is that it opens the possibility of using purely representation-theoretic techniques to understand the multiplicities, as is being pursued in GCT (\eg, \cite{blasiak,BCI,gct4,AdSoSu}). To see how this is possible---that is, how one can discuss multiplicity obstructions without reference to actual test polynomials or modules thereof---we must mention a bit more about the representation theory of $\GL_n$ and $S_n$. Over $\C$, the irreducible representations of these groups have been classified for over 100 years (see, \eg, \cite{fultonHarris}). The equivalence classes of irreducible representations are in bijective correspondence with integer partitions---partitions with at most $n$ parts in the case of $\GL_n(\C)$, and partitions of the number $n$ in the case of $S_n$. The use of partitions enables us to talk about the multiplicities $m_{\lambda}$ and $m_{\lambda}(\overline{\mathcal{C}})$ without reference to any particular (test) module. This is just one of the advantages of the representation-theoretic viewpoint; we discuss two other advantages below.

\subsection{Understanding old lower bounds better (even tight ones!)} \label{sec:old}
In this paper, we show that most arithmetic circuit lower bounds yield separating modules, but typically just one separating module for each lower bound. While this suffices for the lower bound, considering other separating modules that can be used for a given lower bound (or non-separating test modules) may give deeper insight. Indeed, by Fact~\ref{fact:propertiesModules}, this is equivalent to knowing which other invariant properties defined by polynomials can be used (or not) for a lower bound. Understanding which (invariant) properties a complexity class has is surely a task worth undertaking, even for lower bounds that are already tight or as good as we want.

However, trying to understand all such test modules is quite an enormous task. It does not just ask for new proofs of old lower bounds---for example, just asking for a single new separating module for the lower bound---but rather asks to understand, in some sense, \emph{all} possible proofs of a given lower bound. Instead, the difference between separating modules and multiplicity obstructions suggests a more feasible step in this direction which may well be within reach:

\begin{open} \label{open:geometric}
Upgrade the proofs of lower bounds mentioned in this paper from separating modules to multiplicity (or stronger: occurrence) obstructions. 
\end{open}

As a first step towards Open Question~\ref{open:geometric}, which the author hopes to make the subject of future work, we have:

\begin{open} \label{open:labels}
Determine the labels (partitions, see Section~\ref{sec:gct}) of the separating modules in the lower bounds mentioned in this paper.
\end{open}

\subsection{The role of explicitness and constructivity} \label{sec:explicit}
Mulmuley \cite{gctflip} and Williams \cite{williamsConstructive} have both argued for the necessity of constructive methods in proving lower bounds. We can use the representation-theoretic viewpoint to give a further argument for explicitness, albeit a heuristic one. It also allows us to quantify the explicitness or constructivity of known proofs in various ways.

Suppose we are trying to prove $\mathcal{C}_{hard} \not\subseteq \mathcal{C}_{easy}$. In Section~\ref{sec:necessary} and Appendix~\ref{app:necessary} we argue that this is likely to be done using separating modules. If such a separating module exists, it should furthermore be the case that a \emph{random} test module that vanishes on $\mathcal{C}_{easy}$ should not vanish on $\mathcal{C}_{hard}$---and hence be a separating module---for some notion of ``random'' which can probably be made precise. However, to prove the existence of a separating module unconditionally---that is, without assuming the lower bound we are trying to prove---one seems to need a more explicit description of the separating module. This is related to the recent results of Mulmuley \cite{gct5} linking derandomization with algorithms for computational problems in algebraic geometry.

One measure of constructivity is the degree and number (dimension) of test polynomials used. As in the context of Razborov--Rudich \cite{razborovRudich} and Williams \cite{williamsConstructive}, we should expect to measure this degree as a function of something like the size of the truth table of the input polynomials involved. In an algebraic context, we might replace truth table size by the number of monomials. For polynomials of degree $O(n)$ in $\poly(n)$ variables, the number of monomials is $2^{O(n \log n)}$, which is comparable to truth table size.

Another more delicate measure of constructivity is the complexity of verifying that a given test module (perhaps from a specific subset of test modules) is indeed a separating module. This is related to our discussion in Appendix~\ref{app:necessary}.

Using the fact that partitions classify the irreducible representation of $\GL_n$ or $S_n$ over $\C$, we get another measure of constructivity. In general, the dimension of an irreducible representation can be exponential in the bit-size of its corresponding partition, so the partition can serve as a succinct label of an equivalence class of representations. One can then consider the computational complexity of constructing from $0^n$ a partition corresponding to a multiplicity (or occurrence) obstruction for a nonuniform lower bound at input length $n$. Mulmuley conjectures \cite{gctflip} that this construction problem can be solved in $\cc{P}$ for occurrence obstructions in the context of permanent versus determinant and $\cc{NP}$ versus $\cc{P/poly}$. In fact, Mulmuley suggests that finding a polynomial-time algorithm to verify whether a given $\lambda_n$ is the label of an obstruction is a crucial first step towards proving the existence of obstructions unconditionally. This suggests a strengthening of Open Question~\ref{open:geometric}:

\begin{open}
Upgrade the lower bounds mentioned in this paper to multiplicity obstructions where the label $\lambda_n$ of the obstruction at input length $n$ can be computed in $poly(n)$-time. 
\end{open}

Note that resolving Question~\ref{open:labels} would provide natural candidates for labels $\lambda$ that might be multiplicity obstructions. Both Question~\ref{open:labels} and this one seem within reach, especially given the recent occurrence obstructions constructed by B\"{u}rgisser and Ikenmeyer \cite{BI} in the context of matrix multiplication.

The more general question of verification is also interesting:
\begin{open}
For any of the lower bounds mentioned here, what is the complexity of verifying multiplicity or occurrence obstructions? That is, given $\lambda_n$, what is the complexity of verifying that $\lambda_n$ is indeed a multiplicity obstruction?
\end{open}

\subsection{Boolean circuit lower bounds} \label{sec:boolean}
\begin{obs} \label{obs:boolean}
For any Boolean circuit lower bound against a permutation-invariant complexity class---which includes all natural classes, see Section~\ref{sec:invariant}---there is a separating $S_n$-module.
\end{obs}

\begin{proof}
By Fact~\ref{fact:propertiesModulesGeneral} for $S_n$, we only need to argue that the complexity class is defined by test polynomials. As the space of Boolean functions on $n$ variables is finite, every property of $n$-variable Boolean functions is finite and hence defined by test polynomials over $\F_2$. 
\end{proof}

In particular, any nonuniform Boolean lower bound implies the existence of a separating module. 

Despite the fact that the above observation says that separating modules can be used without loss of generality for Boolean circuit lower bounds, we find this observation alone somewhat unsatisfying. However, as with the results of Razborov, Smolensky, and Grigoriev--Razborov over finite fields (see Section~\ref{sec:lowerboundsff}), we believe that many Boolean circuit lower bounds in fact yield separating modules in a very \emph{direct and natural} manner.

Even without having verified this for many known Boolean lower bounds, we can begin to argue why we expect this to be the case. By the discussion in Appendix~\ref{app:necessary}, it is reasonable to expect that lower bounds use properties $\Pi$ which are \emph{naturally} defined by some logical combination of the vanishing of some polynomials and the non-vanishing of other polynomials. We already know that the properties used \emph{can} be defined by the vanishing of some test polynomials; the key here is the \emph{naturality} (in the usual sense of the word, not the Razborov--Rudich sense).

Putting this logical combination into disjunctive normal form, $\Pi$ can be naturally expressed as a union of properties of the form $\Pi_{i} \backslash \Pi_{i}' = \Pi_i \cap \Pi_i'^c$, where each $\Pi_i$ and $\Pi_i'$ is defined by the vanishing of test polynomials and $\Pi_i'^c$ denotes the complement of $\Pi_i'$. Say $\Pi_i'$ is defined by the vanishing of the test polynomials $f_1(x_1, \dotsc, x_n) = \dotsb = f_k(\var{x}) = 0$. Then its complement is most naturally defined by the \emph{non-}vanishing of at least one of the $f_i$. However, the complement $\Pi_i'^c$ can also be defined by the \emph{vanishing} of the single polynomial $\prod_{i=1}^{k} (f_i(\var{x}) - 1)$. Furthermore, by applying $x_i^2 = x_i$, we may take the degree of this single polynomial to be at most $\min\{n, \sum_{i=1}^k \deg(f_i)\}$.

In terms of constructivity, we thus do not lose much by considering $\Pi_i'^c$ as being defined by the vanishing rather than non-vanishing of test polynomials: the single polynomial defining $\Pi_i'^c$ has low degree, and there is only one such polynomial, so the number of polynomials used to define the property also does not increase. 

\begin{rmk}
A similar idea works over any finite field $\F_q$: use $\prod_{0 \neq \alpha \in \F_q} (f_i(\var{x}) - \alpha)$ in place of $f_i(\var{x}) - 1$, reduce by $x_i^q = x_i$, and the resulting degree is at most $(q-1)\min\{n, \sum_i \deg(f_i)\}$.
\end{rmk}

One might argue that using $\prod (f_i(\var{x}) - 1) = 0$ rather than the non-vanishing of some $f_i$ is unnatural, or violates the technique or idea of the lower bound proof that used property $\Pi$. However, if this were really the case, then the lower bound proof would hold for the vanishing/non-vanishing of some $f_i$ as formal polynomials, and hence would work over fields larger than $\F_2$, and in particular would hold over the algebraic closure $\overline{\F}_2$. With the exception of the results mentioned in Section~\ref{sec:lowerboundsff}, we are not aware of Boolean lower bounds that extend to any such fields. In this sense, the use of finiteness in Observation~\ref{obs:boolean} seems less of a kludge to us, and more an essential feature of the current techniques for Boolean circuit lower bounds.

\subsection{Other lower bounds?}
Although we have obviously not considered \emph{all} known lower bounds, we have considered quite a wide cross-section of them in this paper. Of the lower bounds which we actively tried to fit into this framework but have not yet been able to do so, most use heavily machine-based diagonalization. For example, the (non)deterministic time and space hierarchies \cite{hartmanisStearns,cookHierarchy}, uniform lower bounds on the permanent \cite{allenderPerm, allenderGore, koiranPerifelPerm}, time-space trade-offs for $\lang{SAT}$ \cite{fortnow, FLvMV, DvMW, williamsTS, williamsTS2, bussWilliams}, $\cc{\Sigma_2 P} \cap \cc{\Pi_2 P} \not\subseteq \cc{SIZE}(n^k)$ \cite{kannan} and the related result $\cc{MA}_{\cc{EXP}} \not\subseteq \cc{P/poly}$ \cite{BFTMAEXP}. 

\begin{rmk}
Although from one viewpoint Kannan's result rests crucially on the nonuniform circuit-size hierarchy---a counting argument---for the purposes of this discussion the key fact he shows is that a \emph{uniform} $\cc{\Sigma_4 P}$-machine is powerful enough to use the circuit-size hierarchy to diagonalize against $\cc{SIZE}(n^k)$. The same remark applies to the result $\cc{MA}_{\cc{EXP}} \not\subseteq \cc{P/poly}$, as it uses Kannan's result in an essential way.
\end{rmk}

The recent lower bound $\cc{NEXP} \not\subseteq \cc{ACC}^0$ \cite{williams} provides an interesting crucible. It is a nonuniform lower bound against a permutation-invariant Boolean complexity class, hence by Observation~\ref{obs:boolean} there \emph{exists} a separating $S_n$-module proving $\cc{NEXP} \not\subseteq \cc{ACC}^0$. However, the proof uses the nondeterministic time hierarchy in a seemingly crucial way. Extracting a \emph{natural} separating module from Williams's proof may be a first step towards extending the representation-theoretic framework to include uniform lower bounds.

One very interesting technique which we have not yet been able to fit into the representation-theoretic framework and which is only partially uniform comes from Jansen and Santhanam \cite{jansenSanthanam, jansenSanthanamSuccinct}. The key property they use is the existence of $\Z$ hitting sets whose bit descriptions can be encoded by small uniform (or at least succinct \cite{jansenSanthanamSuccinct}) circuits. This combination of algebraic (hitting sets) and Boolean (bit descriptions) frameworks in the same breath makes it difficult to even formulate their proofs in a single algebraic setting, let alone translate them into separating modules.

Finally, Shannon's counting argument \cite{shannon} also seems difficult to put into this representation-theoretic framework. Again, by Observation~\ref{obs:boolean} there exists a separating $S_n$-module for this lower bound. However, finding a natural separating module seems difficult, as Shannon counts the functions in $\mathcal{C}_{easy}$ ($\cc{SIZE}(2^n / n)$ in this case), rather than using some property shared by these functions. This is not necessarily a weakness of the framework however: one of the messages we take from Razborov and Rudich \cite{razborovRudich} is that such simple counting arguments cannot work to prove the strong lower bounds we desire. Indeed, Kadish and Landsberg \cite{kadishLandsberg} point out that getting a lower bound on the determinantal complexity of a generic polynomial is an important first step towards new lower bounds for permanent versus determinant; a lower bound on generic polynomials remains open.

\subsection{Finite fields and positive characteristic} \label{sec:finflds}
There is a mismatch between the current lower bounds over finite fields and the standard techniques of algebraic geometry. The issue is that all the current lower bounds over finite fields that we are aware of depend crucially not just on positive characteristic, but on the size of the field. This means that none of the current lower bounds over finite fields extend to the algebraic closure $\overline{\F}_p$. This is in contrast to the usual approach to finite fields in algebraic geometry, which is (roughly) to first work over their algebraic closures $\overline{\F}_q$ where algebraic geometry and representation theory are nicer and then to pass to the $\F_q$ points\footnote{The $\F_q$-points can be recovered from $\overline{\F}_q$ as the fixed points of the Frobenius map $x \mapsto x^q$, just as $\R$ points can be recovered from $\C$ as the fixed points of the complex conjugation map. The dynamics of the Frobenius map are often very useful.}. In particular, over $\overline{\F}_q$ Hilbert's Nullstellensatz holds and every matrix admits an eigenvector. This process is exactly analogous to (but more complicated than) considering complex solutions, eigenvectors, etc. in order to study equations, matrices, etc. over $\R$. 

As we already mentioned, even if the characteristic is held constant but the field size is allowed to grow at a modest pace with the size of the input, the current lower bounds seem to disappear completely. The essential issue here seems to be that the method of approximations is typically used to ``throw away'' points which are in the \emph{complement} of an algebraic set. Over finite fields, one then argues that these ``erroneous points'' are not too numerous, but over any infinite field, \emph{almost all} points will be ``erroneous,'' as an algebraic set has dimension strictly smaller than that of the ambient space.

It thus seems to us that the limits of our knowledge are not so much in finding lower bounds for depth 3 arithmetic circuits in characteristic zero, as is often stated, but for finding lower bounds for depth 3 arithmetic circuits over any given infinite field, including $\overline{\F}_p$. The chasm at depth 4 \cite{AV, koiranChasm} holds over an arbitrary field, but these observations lead us to wonder:

\begin{open} Is there a chasm at depth 3 over the algebraically closed field $\overline{\F}_p$ for any constant prime $p > 0$?
\end{open}

The current chasm at depth 3 \cite{GKKSchasm} only seems to work in characteristic zero or over a field of (growing) characteristic greater than the degree $d$ of the polynomial, as they use a trick of Fischer \cite{fischer} which requires dividing by $2^{d-1} d!$. 

\section{Most arithmetic circuit lower bounds yield separating modules} \label{sec:lowerbounds}
\newcommand{\lowerbound}[6]{
\vspace{12pt}
\noindent \textbf{#1}

\noindent \textit{Hard function: } #2

\noindent \textit{Complexity class: } #3

\noindent \textit{Lower bound: } #4

\noindent \textit{Invariance: } #5

\noindent \textit{Separating module: } #6 \qed
}

\newcommand{\implication}[5]{
\vspace{12pt}
\noindent \textbf{#1}

\noindent \textit{Assumption: } #2

\noindent \textit{Consequence: } #3

\noindent \textit{Invariance: } #4

\noindent \textit{Separating module implication: } #5 \qed
}

In this section we show how all of the bounds listed in the introduction give separating modules. Rather than recalling all of these proofs and stating a separate proposition for the existence of a separating module for each of these bounds (as in Section~\ref{sec:ex0}), we use a more concise format. Furthermore, we have not included all the results from every paper we consider, but only a representative result from each paper (or sometimes, from each technique). However, we believe that the other results in these papers and using these techniques also yield separating modules.

\subsection{Methods based on partial derivatives}
\ 

\lowerbound{Nisan--Wigderson partial derivatives}
{Elementary symmetric function $e_{n/4,n}$}
{Homogeneous depth 3 arithmetic circuits in characteristic zero}
{Size $2^{\Omega(n)}$ \cite{nisanWigderson}}
{$\F$-linear ($\GL_{n}(\F)$), characteristic zero}
{The $(r+1) \times (r+1)$ minors of the partial derivative matrix $M_f$, as in the proof of Proposition~\ref{prop:example0}.}

\lowerbound{Permanent versus depth 4}
{$\perm_n$}
{Depth 4 $\Sigma \Pi \Sigma \Pi$ arithmetic circuits with bottom fan-in $O(\sqrt{n})$}
{Size $2^{\Omega(\sqrt{n})}$ \cite{GKKS}}
{$\C$-linear ($\GL_{n^2}(\C)$)}
{The outline of the proof of this lower bound is very similar to that for the Nisan--Wigderson lower bound above. However, the key property used here is slightly more complicated. Rather than considering the dimension of the space of partial derivatives $\partial(f)$, they consider the dimension of the space of \emph{shifted} partial derivatives, which are products of polynomials of some degree $\ell$ with the partial derivatives of $f$. Following their notation, we write $\partial^{=k}(f)_{\leq \ell}$ for the space of $k$-th order partial derivatives multiplied by polynomials of degree $\leq \ell$. As in the above case, we build a matrix $\tilde{M}_f$ whose rank is exactly the dimension of $\partial^{=k}(f)_{\leq \ell}$, and then the $r \times r$ minors of this matrix provide the separating module, for appropriately specified $r$, $k$, and $\ell$.

As above, the columns of $\tilde{M}_f$ will be indexed by monomials $\var{x}^e$, and the rows will be indexed by pairs $(\var{x}^d, \partial^{c})$ of a monomial and a partial derivative operator. (Here $c \in \Z_{\geq 0}^n$, and $\partial^c$ denotes $\partial / \partial x_{1}^{c_1} \dotsb \partial x_{n}^{c_n}$.) Then we proceed as in the above case.
}

For this next result, we need a basic fact about sets of polynomials. Given two test modules $V$ and $W$, we define their product $V \cdot W$ as the linear span of the pairwise products of their elements: $V \cdot W = \{ \sum_{i} f_i g_i | f_i \in V, g_i \in W\}$. 

\begin{fact} \label{fact:product}
The disjunction (union) of two invariant properties defined by test polynomials is again an invariant property defined by test polynomials.
\end{fact}

\begin{proof}
Let $V, W$ be test modules. First one verifies that $V \cdot W$ is a test module. Then $V \cdot W$ defines the union of the properties defined by $V$ and $W$. For let $f$ be an input polynomial. If every test polynomial $t \in V$ vanishes at $f$, then so does every test polynomial in $V \cdot W$. Similarly for $W$. Conversely, if some test polynomial $t_1 \in V$ does not vanish at $f$, and some test polynomial $t_2 \in W$ does not vanish at $f$, then $t_1 t_2 \in V \cdot W$ does not vanish at $f$.
\end{proof}

\lowerbound{Multilinear formulas}
{$\det_n$ or $\perm_n$}
{(Syntactic) multilinear formulas in characteristic zero}
{Size $\Omega(n^{\log n})$ \cite{razPerm}}
{permutation ($S_n)$}
{Raz combines the above ideas on $\dim \partial(f)$ with random restrictions, making the separating module here slightly more complicated than in the above examples. Raz explicitly defines a matrix of partial derivatives, similar to that in the above two examples, which he also denotes $M_f$. The random restrictions Raz uses (see \cite[\S 5]{razPerm}) take the form $\rho(x_i, x_j, x_k, x_\ell) = (1,1,y_m, z_m)$, where the $i,j,k,\ell$ used are of a particular form, and the image may be re-ordered in one of two possible ways. In particular, for each input length $n$ there are only finitely many such restrictions to consider. 

He then shows a lower bound on $\rk M_{\det(\rho(X))}$ and $\rk M_{\perm(\rho(X))}$ under \emph{any} such restriction $\rho$, and using a probabilistic argument shows that there \emph{exists} a restriction making $\rk M_{f(\rho(X))}$ small when $f$ is computed by a multilinear formula of size $n^{o(\log n)}$. Hence the property he is using is that there \emph{exists} a restriction $\rho$ as in his \S 5 which makes $\rk M_{f(\rho(X))} \leq r$ for appropriately chosen $r$. 

For a given restriction $\rho$, we get a test $S_n$-module $V_{\rho}$ consisting of the $(r+1) \times (r+1)$ minors of $M_{f(\rho(X))}$. This test $S_n$-module vanishes if and only if $\rk M_{f(\rho(X))} \leq r$. The separating module is then the product over all (finitely many) $\rho$ of the $V_{\rho}$ (cf. Fact~\ref{fact:product}).
}

\begin{rmk}
Although bounding the rank of a matrix of partial derivatives is linear-invariant, the property of being multilinear is \emph{not} linear-invariant, though it is permutation-invariant. Hence, despite using a bound on the dimension of partial derivatives,  it was to be expected that at some point in the proof a property would be used that was only permutation-invariant and not linear-invariant. Although Raz uses multilinearity elsewhere in his proof, even in the brief outline above we see that the type of random restrictions used is only permutation-invariant, and not linear-invariant.
\end{rmk}

\subsection{Methods using properties of (semi-)algebraic varieties} \label{sec:varieties}
For methods such as the degree bound \cite{strassenDegree, baurStrassen} and the connected components technique \cite{benOr}, the most natural input objects to use are themselves (semi-)algebraic varieties. In other words, we need to replace the input space $\Poly^d(\var{x})$ with a space whose points correspond to varieties. Such spaces have been constructed in (semi-)algebraic geometry, but their construction is not as elementary as in the above results. In both cases the basic idea is that the input objects will in fact be systems of equations (which, in turn, define algebraic sets), and the test variables are then the coefficients of these systems of equations. 

Surprisingly, the use of these ``parameter spaces of algebraic sets'' makes putting these results into the representation-theoretic viewpoint technically more complicated than the above results, despite the fact that these bounds were discovered considerably earlier. 

\lowerbound{The degree bound}
{Computing all elementary symmetric functions $e_{1,n}, \dotsc, e_{n,n}$ together}
{Arithmetic circuits over an infinite field}
{Size $\Omega(n \log n)$ \cite{strassenDegree}}
{$\F$-affine ($\AGL_{n}(\F)$), $\F$ infinite}
{The key property used here is the degree of a projective algebraic set. Although the degree has a nice geometric definition (in characteristic zero), here we recall the algebraic definition as it lends itself more readily to the definition of the separating module. Let $V$ be an algebraic subset of $\PP(\F^{n})$, and let $I \subseteq \F[x_{1}, \dotsc, x_{n}]$ be the homogeneous ideal of all polynomials that vanish on $V$. In particular, $I$ can be written as the direct sum $\bigoplus_d I_d$ of its homogeneous subsets $I_d$, which consist of those polynomials in $I$ of degree exactly $d$. The \emph{Hilbert function} of $I$ is then $h_I(d) \defeq \dim_{\F} I_d$. Hilbert showed (see, \eg, \cite[\S 9.3]{CLO} or \cite[Thm.~1.11]{eisenbud}) that for all sufficiently large $d$, $h_I(d)$ agrees with a polynomial $p_I(d)$, which is referred to as the \emph{Hilbert polynomial} of $I$ or $V$. The \emph{degree} of $V$ is then the leading coefficient\footnote{In an unfortunate twist of terminological fate, it turns out that the dimension of $V$ in the usual sense is equal to the degree of its Hilbert polynomial.} of the Hilbert polynomial $p_I(d)$.

For the input space, we may use either the Chow variety (see, \eg, \cite[Ch.~3, \S 7]{danilov}) or the Hilbert scheme (see, \eg, \cite{grothendieckHilbert}). The Chow variety is essentially the ``space of projective algebraic sets,'' and the Hilbert scheme is essentially the ``space of homogeneous ideals in $\F[x_1, \dotsc, x_n]$.'' The Chow variety is in fact a disjoint union over pairs $(d, D)$ of the variety of projective algebraic sets of degree $d$ and dimension $D$. Similarly, the Hilbert scheme is the disjoint union over Hilbert polynomials $p(\cdot)$ of the scheme of homogeneous ideals with Hilbert polynomial $p_I = p$. In either case, showing that two varieties have different degrees then amounts to showing that these varieties, as points in the space of varieties, live in different connected components of the Chow variety or Hilbert scheme. 

Finally, being in a given component of a variety (or scheme) is defined by the vanishing of some (test) polynomials. As the Hilbert polynomial, and in particular the degree and dimension, is an affine invariant of a projective algebraic set, the components of the Chow variety and Hilbert scheme are also affine-invariant. Hence, by the analog of Fact~\ref{fact:propertiesModules} for affine invariance, there is a separating module.
}

\lowerbound{Algebraic decision trees for sorting}
{Element distinctness (note that element distinctness reduces to sorting)}
{Real semi-algebraic decision trees}
{Depth $\Omega(n \log n)$ \cite{benOr}}
{$\R$-affine ($\AGL_{n}(\R)$)}
{The key property used here is the number of connected components of a semi-algebraic variety---that is, a subset of $\R^n$ defined by a collection of polynomial equalities and inequalities. The number of connected components is clearly affine-invariant; we recall here how Hardt's Triviality Theorem implies that it is in fact defined by a collection of test polynomial equalities and inequalities. The use of inequalities here is unavoidable: see Remark~\ref{rmk:inequalities} below.

A special case of Hardt's Triviality Theorem \cite{hardt} (see, \eg, \cite[\S 5.8]{basuPollackRoy} for a textbook treatment) says that for any continuous semi-algebraic map $\pi \colon S \to \R^N$ from a semi-algebraic set $S \subseteq \R^n$, there is a finite partition of $\R^N$ into semi-algebraic sets $\R^N = \bigcup_{i=1}^{k} T_i$ such that for each $i$ and every $x \in T_i$, $T_i \times \pi^{-1}(x)$ is semi-algebraically homeomorphic to $\pi^{-1}(T_i)$. In particular, this implies that for each $i$, if $x, y \in T_i$ then $\pi^{-1}(x)$ and $\pi^{-1}(y)$ have the same number of connected components.

Now, consider a collection of polynomial equalities and inequalities of degree $\leq d$ in $n$ variables $x_1, \dotsc, x_n$:
\begin{equation} \label{eqn:semialgebraic}
\begin{array}{rclcrcl}
\sum_{e} a_{1,e} \var{x}^{e} + a_{1} & = & 0, & \dotsc, & 
\sum_{e} a_{m,e} \var{x}^{e} + a_{m} & = & 0 \\
\sum_{e} a_{m+1,e} \var{x}^{e} + a_{m+1} & \geq & 0, & \dotsc, & 
\sum_{e} a_{m+s,e} \var{x}^{e} + a_{m+s} & \geq & 0 \\
\sum_{e} a_{m+s+1,e} \var{x}^{e} + a_{m+s+1} & > & 0, & \dotsc, & 
\sum_{e} a_{h,e} \var{x}^{e} + a_{h} & > & 0
\end{array}
\end{equation}
We may consider the $a_{i,e}$s and $a_{i}$s as variables rather than constants; suppose in total there are $N$ such variables. Then the $x_i$s are coordinates on $\R^n$ and the $a_{i,e}$s are coordinates on $\R^N$. Equations (\ref{eqn:semialgebraic}) thus define a semi-algebraic subset $S \subseteq \R^n \times \R^N$. Let $\overline{\pi}\colon \R^n \times \R^N \to \R^N$ be the projection onto the second factor, and let $\pi\colon S \to \R^N$ be the restriction of $\overline{\pi}$ to $S$. For any given numerical values $\var{a} \in \R^N$, let $V_{\var{a}} \subseteq \R^n$ denote the semialgebraic subset defined by (\ref{eqn:semialgebraic}). Then $\pi^{-1}(\var{a}) = V_{\var{a}} \times \{ \var{a} \} \cong V_{\var{a}}$ (where $\cong$ here denotes semialgebraic homeomorphism).

Finally, by Hardt's Triviality Theorem, there is a semialgebraic partition $\R^N = \bigcup_{i=1}^{k} T_i$ such that for any $\var{a}$ and $\var{a}'$ in the same $T_{i}$, $V_{\var{a}}$ and $V_{\var{a}'}$ have the same number of connected components. Hence, the collection of equations of the form (\ref{eqn:semialgebraic}) that define a semialgebraic variety with $c$ connected components is the semi-algebraic set $\bigcup \{ T_i | \pi^{-1}(\var{a}) \text{ has $c$ connected components for all } \var{a} \in T_i \}$. As the property of having $c$ connected components is invariant under affine transformations of the $x_i$s ($\AGL_n(\R)$), this union of $T_i$s is also affine-invariant (under the induced action of the same $\AGL_n(\R)$, not under the larger $\AGL_N(\R)$), and hence is defined by some affine-invariant collection of equalities and inequalities (by an analog of Fact~\ref{fact:propertiesModules}).
}

\begin{rmk} \label{rmk:inequalities}
The use of inequalities here is necessary. The vanishing of some test polynomials would not suffice, even when the semi-algebraic variety is defined only by equalities. This can be seen even in the simple case of the number of connected components defined by a quadratic: over $\R$ the number of connected components of the algebraic set $\{x \in \R | ax^2 + bx + c = 0\}$ is zero if and only if $b^2 - 4ac < 0$ and is at most one if and only if $b^2 - 4ac \leq 0$. The set $\{(a,b,c) | b^2 - 4ac < 0\}$ is not defined by the vanishing of some polynomials, for it has dimension 3, but the only 3-dimensional subset of $\R^3$ defined by the vanishing of polynomials is $\R^3$ itself. Hence inequalities are necessary.
\end{rmk}

\begin{rmk} 
Note that the above lower bound implies the same lower bound for decision trees for element distinctness (and sorting) over $\C$. However, over $\C$ the connected components argument does not work directly, because semi-algebraic varieties over $\C$ tend to have fewer connected components than over $\R$. In particular, the semialgebraic variety corresponding to element distinctness over $\C$ has just a single connected component. Hence although the lower bound holds over $\C$, we would still only get a separating $\AGL_n(\R)$-module. 
\end{rmk}

\lowerbound{Algebraic decision trees for $k$-equals}
{$k$-equals (are at least $k$ of the inputs equal?)}
{Real semi-algebraic decision trees}
{Depth $\Omega(n \log (n/k))$ \cite{yao97}}
{$\R$-affine ($\AGL_{n}(\R)$)}
{The key property used here is a lower bound on \emph{any} Betti number, rather than just the number of connected components (=the $0$-th Betti number). As the Betti numbers are invariant under homeomorphism, essentially the same argument as above using Hardt's Triviality Theorem works for this result.
}

\subsection{The method of approximations over finite fields} \label{sec:lowerboundsff}
Here we give two representative examples of how results that use the method of approximation for circuits over finite fields yield separating modules. Results using similar properties, such as those of Grigoriev--Razborov \cite{GR}, should similarly yield separating modules.

\lowerbound{Razborov--Smolensky}
{$\text{MOD}_3$}
{$\cc{AC}^0[2]$}
{Exponential size \cite{razborov, smolensky}}
{$\F_2$-affine ($\AGL_{n}(\F_2)$)}
{Every $\cc{AC}^0[2]$ circuit computes a polynomial function over $\F_2$, so we use $\Omega^{d,n}_{\F_2} \defeq \Poly^d_{\F_2}(x_1, \dotsc, x_n) / \langle x_1^2 = x_1, \dotsc, x_n^2 = x_n \rangle$ as the space of input functions (using $\Omega$ we follow Smolensky's notation). Note that here we consider two functions equal if they are equal when evaluated on all $\F_2$ points. In other words, we are considering \emph{functions} on $\F_2$, rather than formal polynomials whose coefficients are in $\F_2$. Every function over $\F_2$ can be represented by a unique multilinear polynomial; when we refer to $\text{MOD}_3$ we mean its corresponding $\F_2$-multilinear polynomial. 

Fix a depth $k$ and a constant $\lambda$. For our purposes, the key property used here is:
\begin{quotation}
There exists a subset $\Gamma \subseteq \F_2^n$ (for ``good'') of size at least $2^n - 2^{n-r}$ such that $f$ agrees with a polynomial of degree $\leq (2\lambda r)^k$ on the points in $\Gamma$. 
\end{quotation}
Smolensky \cite[Lem.~2]{smolensky} shows that this holds for any function computed by a depth $k$ circuit with parity gates for $r=o(n^{1/2k})$, but not for $\text{MOD}_3$. This condition is clearly $\GL_n(\F_2)$-invariant.

For any $\Gamma \subseteq \F_2^n$, let $I_{\Gamma}$ be the ideal of polynomials that vanish everywhere on $\Gamma$. When we mod out the space of functions by $I_{\Gamma}$, this is the same as only considering the values a function takes on $\Gamma$. Then $f$ agrees with a polynomial of degree $\leq d=(2\lambda r)^k$ on the points in $\Gamma$ if and only if all of the coefficients of monomials of degree $> d$ of $f \pmod{I_{\Gamma}}$ vanish. As the map $\Omega^{d,n} \to \Omega^{d,n} / I_{\Gamma}$ is linear, the coefficients of $f \pmod{I_{\Gamma}}$ are linear combinations of the coefficients of $f$, and we are asking that certain such linear combinations vanish. Let $T_{\Gamma}$ be the test module consisting of these linear combinations. Finally, for an appropriate choice of $r$, by Fact~\ref{fact:product}, $\prod_{\Gamma} T_{\Gamma}$ is the desired separating module, where the product is taken over all (finitely many) subsets $\Gamma \subseteq \F_2^n$ of size $\geq 2^n - 2^{n-r}$.
}

\lowerbound{Depth 3 arithmetic circuits over finite fields}
{Determinant or permanent}
{Depth 3 arithmetic circuits over the finite field $\F_{q}$}
{Exponential size \cite{GK}}
{$\F_{q}$-linear ($\GL_{n}(\F_{q})$)}
{As above, the key property here will use an existential quantifier over some finite collection of subsets $S$ of $\F_q^n$, which will turn into a big product of test modules over all possible choices for $S$. Beyond that, the condition here is quite a bit more complicated than above.

Here, we work in the space of formal polynomials over $\F_q$, namely $\Poly^d_{\F_q}(x_{11}, x_{12}, \dotsc, x_{nn})$. To describe the key property we introduce some notation. Given $\sigma \in \GL_n(\F_q)$ and any function $f = f(X)$, we denote $f(\sigma X)$ by $f^{\sigma} = f^{\sigma}(X)$. For any set $F$ of functions, write $F^{\sigma} = \{f^{\sigma} | f \in F\}$. Let $\partial^{\leq r}(f)$ denote the linear span of all the partial derivatives of $f$ of order $\leq r$. Finally, combining these notations, we have $\partial^{\leq r}(f)^{\sigma} = \{ g^\sigma | g \in \partial^{\leq r}(f)\}$. 

The key property of a function $f \in \Poly^d_{\F_q}(x_1, \dotsc, x_n)$ is then, for appropriate choices of all the parameters involved, that there exists a subset $S \subseteq \GL_n(\F_q)$ of size $\leq s$ such that
\begin{quotation}
there is a function $g(X)$ in the intersection $\bigcap_{\sigma \in S} \partial^{\leq r}(f)^{\sigma}$ such that $g(A) = 0$ for all $A \in \GL_n(\F_q)$.
\end{quotation}
Again, this property is readily seen to be $\GL_n(\F_q)$-invariant. Let us verify that it is defined by test polynomials. For now, fix a subset $S \subseteq \GL_n(\F_q)$. For each $\sigma \in S$, we compute a linear basis of $\partial^{\leq r}(f)^{\sigma}$. The coefficients of each such basis function will be linear combinations of the coefficients of $f$ (=test variables). This follows from the usual fact about partial derivatives, and the fact that for any $\sigma \in \GL_n(\F_q)$ and any function $h$, the coefficients of $h^{\sigma}$ are linear combinations of the coefficients of $h$. Next, we take the intersection over all $\sigma \in S$ of these subspaces. Again, a linear basis for the resulting intersection will consist of polynomials whose coefficients are linear combinations of the test variables. Let us denote this intersection $\Lambda$. 

Now observe that the collection of all $g$ such that $g(A) = 0$ for all $A \in \GL_n(\F_q)$ is an ideal $I$ in the space of polynomials (of degree $\leq d$ for some $d$), and in particular is a linear subspace thereof. Then the property is satisfied exactly if $I \cap \Lambda \neq 0$. The system of linear equations defining $I \cap \Lambda$ has coefficients which are either linear combinations of the coefficients of $f$ (coming from the equations defining the linear space $\Lambda$) or constants (coming from the equations defining $I$). If this system of equations had the same number of variables as equations we could require that just the $n \times n$ determinant of the system vanishes. As the system is likely to have more equations than variables, we must require that all the $n \times n$ minors of this system vanish. These $n \times n$ minors form a test module $T_S$, and then, as above, the separating module is $\prod_S T_S$, where the product is over all $S$ of appropriate size.
}

\begin{rmk}
Aside from the more obvious uses of finiteness (not just finite characteristic) in the above proofs, in the Grigoriev--Karpinski proof, the property they use becomes vacuous over any infinite field $\F$: the only polynomial in $n^2$ variables that vanishes everywhere on $\GL_n(\F)$ is the zero polynomial. For further discussion of these issues see Section~\ref{sec:finflds}.
\end{rmk}

\subsection{Results previously known to give separating modules}
\

\lowerbound{Permanent versus determinant}
{$\perm_n$}
{Linear projections of $\det_m$}
{$m \geq n^2/2$ \cite{mignonRessayre}; also border determinantal complexity $n^2 / 2$ \cite{LMR}}
{$\C$-linear ($\GL_{m^2}(\C)$)}
{The key property used by Mignon and Ressayre \cite{mignonRessayre} is the rank of the Hessian matrix of a function. Recall that the Hessian of a function $f(x_1, \dotsc, x_n)$ is the $n \times n$ matrix $\Hess(f)$ whose $(i,j)$ entry is the second partial derivative $\partial^2 f / \partial x_i \partial x_j$. They show a lower bound on $\rk \Hess(\perm)$ and an upper bound on $\rk \Hess(\det)$. Note that the entries of $\Hess(\det)$ are themselves functions; the upper bound on $\rk \Hess(\det)$ that they prove does not hold at all matrices $X$, but only at those matrices where $\det(X) = 0$. This is enough for them to prove the lower bound, but makes it complicated to extract a separating module from their proof. 

If the upper bound held for all $X$, then the minors of the Hessian matrix would span a separating module, as in the Nisan--Wigderson partial derivatives technique above. Instead, the condition they use is that $\det(X)$ divides the $r \times r$ minors of $\Hess(\det)$ (for $r=2n+1$). Landsberg, Manivel, and Ressayre \cite{LMR} find polynomial equations that vanish exactly on the pairs of polynomials $(f, g)$ such that $f$ divides $g$ (amongst other achievements), resolving a surprisingly old question in algebraic geometry. They then construct a separating module by using these equations with $f=\det$ and $g$ the minors of $\Hess(\det)$.}

\lowerbound{Matrix multiplication}
{$n \times n$ matrix multiplication}
{Bilinear circuits in characteristic zero}
{Border rank $\geq \frac{3}{2} n^2 - o(n^2)$ \cite{BI}}
{$\F$-linear ($\GL_{n^2}(\F) \times \GL_{n^2}(\F) \times \GL_{n^2}(\F)$, characteristic zero}
{B\"{u}rgisser and Ikenmeyer \cite{BI} explicitly construct separating modules yielding this lower bound (in fact, they construct ``occurrence obstructions,'' see Section~\ref{sec:gct} below for the definition).}

\section{Relations between lower bounds yield relations between separating modules} \label{sec:relations}

\implication{Baur--Strassen: computing partial derivatives \cite{baurStrassen}}
{Computing $(\partial f/\partial x_{1}, \dotsc, \partial f / \partial x_{n})$ requires arithmetic circuits of size $s$}
{Computing $f$ requires arithmetic circuits of size $s/3$}
{$\F$-linear ($\GL_n(\F)$), any infinite field}
{Let $\varphi$ be the map from $\Poly^d(x_1, \dotsc, x_n)$ to the Chow variety or Hilbert scheme (see The Degree Bound above), defined as follows. $\varphi(f)$ is the variety (ideal) defined by $\langle \partial f / \partial x_1, \dotsc, \partial f / \partial x_n \rangle$. Recall that $A \in \GL_n(\F)$ acts on the Hilbert scheme by taking the ideal $\langle g_1(\var{x}), \dotsc, g_k(\var{x}) \rangle$ to the ideal $\langle g_1(A\var{x}), \dotsc, g_k(A\var{x}) \rangle$; let us denote the latter by $A \cdot \langle g_1(\var{x}), \dotsc, g_k(\var{x}) \rangle$. Similarly, $A \in \GL_n(\F)$ acts on $\Poly^d(\var{x})$ by sending $f(\var{x})$ to $f(A\var{x})$. Then $\varphi$ is $\GL_n(\F)$-equivariant, in that
\begin{eqnarray*}
\varphi(f(A\var{x})) & = & \left\langle \sum_{j} a_{1j} \left( \frac{\partial f}{\partial x_j}\right)(A\var{x}), \dotsc, \sum_{j} a_{nj} \left( \frac{\partial f}{\partial x_j} \right)(A\var{x}) \right\rangle \\
 & = & \left\langle \left( \frac{\partial f}{\partial x_1}\right)(A\var{x}), \dotsc, \left( \frac{\partial f}{\partial x_n}\right)(A\var{x}) \right\rangle \\
 & = & A \cdot \varphi(f(\var{x})).
\end{eqnarray*}
If $T$ is a test module which vanishes on $\{ \varphi(g) | \varphi(g) \text{ has arithmetic circuits of size } \leq s\}$, but not on $\varphi(f)$, then $\varphi_{*}(T) \defeq \{t \circ \varphi | t \in T \}$ is a vector space of test polynomials which vanishes at all $g \in \Poly^d(\var{x})$ that have circuits of size $\leq s/3$, but not at $f$. The $\GL_n(\F)$-equivariance of $\varphi$ implies that $\varphi_{*}(T)$ is in fact a test $\GL_n(\F)$-module.
}

\implication{Tensor rank to formula size \cite{razTensor}}
{$t_n \in (\F^{n})^{\otimes r(n)}$ has tensor rank $\geq n^{r(n)(1 - o(1))}$ for some $\omega(1) \leq r(n) \leq O\left(\frac{\log n}{\log \log n}\right)$}
{The polynomial $f_n$ which is the symmetrization of $t_n$ requires super-polynomial size arithmetic formulas. Also, by the completeness of the permanent, $\perm_n$ requires super-polynomial size arithmetic formulas (attributed to Yehudayoff, \cite[Footnote~2]{razTensor})}
{$\F$-linear ($\GL_{n}(\F)$), $\F$ arbitrary}
{Raz uses the standard symmetrization map from tensors $(\F^{n})^{\otimes r}$ (we think of these as degree $r$ homogeneous noncommutative polynomials) to $\Poly^r(x_1, \dotsc, x_n)$. In particular, to show an arithmetic formula size lower bound on some $f_n \in \Poly^r(\var{x})$, it suffices to show a tensor rank lower bound on \emph{any} noncommutative version $t_n$ of $f_n$ (that is, $f_n$ is the result of symmetrizing $t_n$). In particular, we are free to use the standard embedding (NB: in the opposite direction compared to the above) $\varphi\colon\Poly^r(\var{x}) \hookrightarrow (\F^n)^{\otimes r}$, which takes the monomial $x_{i_1} \dotsc x_{i_r}$ to the tensor $\frac{1}{r!}\sum_{\pi \in S_r} x_{i_{\pi(1)}} \otimes \dotsb \otimes x_{i_{\pi(r)}}$. Raz's results imply that the image, under $\varphi$, of the set of polynomials that have small formulas is contained in the set of tensors of low tensor rank. It is a standard fact from multi-linear algebra that the embedding $\varphi$ is $\GL_n(\F)$-equivariant (see the Baur--Strassen implication above). Hence, if a test module $T$ is used to show a lower bound on the tensor rank (and hence, border rank, see Section~\ref{sec:border}) of some $\varphi(f)$, then $\{ t \circ \varphi | t \in T\}$ is a test module which implies the stated lower bound on the arithmetic formula size of $f$.
}

\implication{Chasm at Depth 4 \cite{AV, koiranChasm}}
{$f$ requires depth 4 arithmetic circuits of size $2^{\omega(\sqrt{n}\log^2 n)}$}
{$f$ requires arithmetic circuits of super-polynomial size}
{$\F$-affine ($\AGL_{n}(\F)$), $\F$ arbitrary}
{They show that the set of functions computable by arithmetic circuits of polynomial size is contained in the set of functions computable by depth $4$ circuits of size $2^{O(\sqrt{n} \log^2 n)}$. Hence, if a separating module vanishes on the latter set, it also vanishes on the former. 

}

\implication{Chasm at Depth 3 \cite{GKKSchasm}}
{$f$ requires depth 3 arithmetic circuits of size $2^{\omega(\sqrt{n}\log^{3/2} n)}$}
{$f$ requires arithmetic circuits of super-polynomial size}
{$\F$-affine ($\AGL_{n}(\F)$), characteristic zero or characteristic $> \deg f$}
{Same as above, but with different bounds and not over arbitrary fields. See Section~\ref{sec:finflds} for a discussion of this issue.}

\implication{Matrix rigidity to linear circuits \cite{valiantRigidity}}
{The $n \times n$ matrix $A_n$ has rigidity $R_{A_n}(n/2) \geq \Omega(n^{1 + \varepsilon})$}
{The linear function $\var{x} \mapsto A_n\var{x}$ does not have linear circuits of simultaneous size $O(n)$ and depth $O(\log n)$}
{permutation ($S_{n} \times S_n$)}
{Here the ambient (input) space is the space $M_n(\F)$ of $n \times n$ matrices. Valiant \cite[Cor.~6.3]{valiantRigidity} showed the set of matrices $A_n$ whose associated linear functions $x \mapsto A_n x$ can be computed by linear circuits of size $O(n)$ and depth $O(\log n)$ (simultaneously) is contained in the set of matrices of low rigidity. Hence any test module which vanishes on the set of matrices with low rigidity but not on some matrix $A$ will also vanish on the set of matrices that can be computed in size $O(n)$ and depth $O(\log n)$ by linear circuits.

As the concept of rigidity involves the \emph{number of entries} of a matrix that must be changed to drop its rank, this concept is only permutation-invariant---we may multiply $A_n$ on the left and right by permutation matrices without affecting its rank or rigidity. We note that, despite the fact that the non-rigid matrices do not form an algebraic set, some of the most successful results on matrix rigidity to date use the algebro-geometric approach (essentially, test polynomials) \cite{lokamRigidity} (see also \cite{landsbergRigidity} for more on the geometry).
}

\section*{Acknowledgments}
The author would like to thank Scott Aaronson, Eric Allender, Saugata Basu, Tom Church, Klim Efremenko, Kaveh Ghasemloo, J. M. Landsberg, Ketan Mulmuley, Toni Pitassi, Peter Scheiblechner, Chris Umans, Alasdair Urquhart, Ryan Williams, and Yiwei She for useful discussions. In particular Williams suggested the terminology ``input polynomial vs. test polynomial,'' and Shipman and Church helped come up with the name ``separating module.'' The author also had several useful conversations with Shipman regarding the Hilbert scheme, used to show the degree bound fits into this framework. Basu pointed the author to Hardt's Triviality Theorem, used to show the connected components sorting lower bound fits into this framework. Discussions with Scheiblechner were useful for Section~\ref{sec:boolean}. The author had useful discussions with Ghasemloo and Pitassi regarding whether and how the Razborov--Smolensky result fits into this framework. The author would also like to thank Amir Yehudayoff for his wonderfully clear talk on multilinear lower bounds at Dagstuhl. Some of these conversations were facilitated by the Dagstuhl seminar on Algebraic and Combinatorial Methods in Computational Complexity in October 2012 (thanks to the organizers Manindra Agrawal, Thomas Thierauf, and Chris Umans for the invitation), and by the Exploring the Limits of Computation Tokyo Complexity Workshop in March 2013 (thanks to the organizers Osamu Watanabe and Jun Tarui for the invitation and support). This paper would not have been possible without many useful conversations with Ketan Mulmuley over several years. Toni Pitassi provided very helpful comments on a draft. The author would also like to thank Allan Borodin for his support, both material and otherwise, part of which came by way of Borodin's NSERC Grant \#482671.

\appendix
\section{Proof of the correspondence between invariant properties and test modules} \label{app:invariant}

\begin{fact}[Generalized restatement of Fact~\ref{fact:propertiesModules}] \label{fact:propertiesModulesGeneral}
Let $G$ be any finite or algebraic group acting algebraically on an input space. There is a many-to-one correspondence between test $G$-modules and $G$-invariant properties defined by the vanishing of test polynomials.
\end{fact}

For readers unfamiliar with algebraic geometry, we note that $\GL_n(\F)$ and $\AGL_n(\F)$ are both algebraic groups. All of the situations considered in this paper satisfy the hypotheses above. 

For readers familiar with algebraic geometry but perhaps not with algebraic groups: an algebraic group is an algebraic set that is also a group, and where the multiplication map $G \times G \to G$ and inversion map $G \to G$ are both algebraic maps. In particular, all finite groups are algebraic. (If you are concerned that $\GL_n$ is a Zariski-open subset of $\F^{n^2}$, consider $\GL_n$ as instead the algebraic set $\{(A, \frac{1}{\det A}) | A \text{ invertible}\} \subseteq \F^{n^2 + 1}$.) An action of $G$ on an algebraic set $V$ is algebraic if the action map $G \times V \to V$ is algebraic.

\begin{proof}
Let $V$ denote the input space (input polynomials, matrices, etc.), and suppose that $T$ is a test $G$-module with basis $t_1, \dotsc, t_k$. Let $\Pi_T$ denote the corresponding property, namely $\Pi_T = \{ v \in V | t(v) = 0 \forall t \in T\}$. $\Pi_T$ is defined by test polynomials (namely, those in $T$). To see that $\Pi_T$ is $G$-invariant, suppose that $v \in \Pi_T$ and $g \in G$, and consider the point $gv$. By the defining property of test $G$-module, if $t(\var{x}) \in T$, then $t(g\var{x}) \in T$ for all $g \in G$. Let $t'(\var{x}) = t(g\var{x})$. As $t' \in T$ and $v \in \Pi_T$, we have $t'(v) = 0$ by the definition of $\Pi_T$. But then $0=t'(v) = t(gv)$, as desired. Hence $\Pi_T$ is a $G$-invariant property defined by test polynomials.

Conversely, suppose that $\Pi \subseteq V$ is a $G$-invariant property defined by test polynomials. By Hilbert's Basis Theorem, $\Pi$ is defined by the vanishing of only finitely many test polynomials, say $t_1, \dotsc, t_k$. If $G$ is finite, then it is clear that the collection of polynomials $GT \defeq \{ t_i(g(\var{x})) | 1 \leq i \leq k, g \in G\}$ is finite. If $G$ is algebraic, then it is a standard fact from algebraic geometry that the linear span of $GT$ is finite-dimensional, even though $G$ itself may be infinite. It is clear from the construction that $GT$ is a test $G$-module. It remains to show that $\Pi$ is exactly the set of input points on which $GT$ vanishes. Let us denote the latter set by $\Pi_{GT}$. By the previous direction, $\Pi_{GT}$ is $G$-invariant.

We will show that for arbitrary $\Pi$ defined by test polynomials in $T$ (not necessarily $G$-invariant), $\Pi_{GT}$ is the unique maximum $G$-invariant subset of $\Pi$. Hence, if $\Pi$ itself is $G$-invariant, then $\Pi=\Pi_{GT}$. Suppose $\Pi'$ is a $G$-invariant subset of $\Pi$. In particular, every test polynomial $t \in T$ vanishes on every $v \in \Pi'$. We must show that for arbitrary $g$, $t(g\var{x})$ also vanishes on every $v \in \Pi'$. As $\Pi'$ is $G$-invariant, $v \in \Pi'$ implies that $gv \in \Pi'$ for every $g \in G$. Hence $t(gv) = 0$ for every $v \in \Pi'$. Thus $\Pi' \subseteq \Pi_{GT}$. As this holds for arbitrary $G$-invariant subsets $\Pi'$ of $\Pi$, $\Pi_{GT}$ is the unique maximum $G$-invariant subset of of $\Pi$, and thus is equal to $\Pi$ if $\Pi$ itself is $G$-invariant.
\end{proof}

It is clear that the map sending a test $G$-module $T$ to the property $\Pi_{T}$ is well-defined, and hence is at worst many-to-one. Over an algebraically closed field, Hilbert's Nullstellensatz implies that two test $G$-modules $T_1$ and $T_2$ define the same property $\Pi$ if and only if they generate the same radical ideal. Hence, we cannot expect this map to be one-to-one.

\section{The utility of separating modules} \label{app:necessary}
In Section~\ref{sec:invariant} we argued that invariant properties can be used to prove lower bounds without loss of generality. In Section~\ref{sec:constructible} we argued that for all naturally occurring nonuniform complexity classes $\mathcal{C}$, $\mathcal{C}_n$ is constructible, and furthermore is typically the image of some simple algebraic map from some $\F^N$. We now give a heuristic argument that the easiest way to prove a lower bound against such sets is by using a test polynomial, and hence, for invariant classes, a separating module. Even when the use of separating modules is not formally necessary, it thus helps illuminate any (constructible) nonuniform complexity class.

If $\mathcal{C}_n$ is closed, then test polynomials are necessary and sufficient to prove $f_{hard,n} \notin \mathcal{C}_n$ (see Section~\ref{sec:constructible}). For the sake of discussion, suppose that $\mathcal{C}_n$ is not closed, but is the next simplest kind of constructible set: $\mathcal{C}_n$ is the difference $\mathcal{A}_n \backslash \mathcal{B}_n$ of two closed sets $\mathcal{A}_n, \mathcal{B}_n$. By what is essentially disjunctive normal form, every constructible set is a union of such differences. 

Without loss of generality, we may assume that $\mathcal{A}_n = \overline{\mathcal{C}_n}$ is the Zariski-closure of $\mathcal{C}_n$, and that $\mathcal{B}_n \subseteq \mathcal{A}_n$. Equivalently, $\mathcal{B}_n = \overline{\mathcal{C}_n} \backslash \mathcal{C}_n$ is the boundary of $\mathcal{C}_n$. 

Two approaches to show $f_{hard,n} \notin \mathcal{C}_n$ immediately suggest themselves: (1) show that $f_{hard,n} \notin \mathcal{A}_n = \overline{\mathcal{C}_n}$; or (2) show that $f_{hard,n} \in \mathcal{B}_n$. As $\mathcal{B}_n = \overline{\mathcal{C}_n} \backslash \mathcal{C}_n$ might be complicated, a third approach is (3) to find a closed set $\mathcal{D}_n$ containing $f_{hard,n}$ such that $\mathcal{D}_n$ is disjoint from $\mathcal{C}_n$. Each of these approaches of course requires some insight: in general, (1) requires finding a test polynomial with the desired properties, (2) requires finding all test polynomials that vanish on $\mathcal{B}_n$, or at least a set of test polynomials whose vanishing defines $\mathcal{B}_n$, and (3) requires finding the set $\mathcal{D}_n$ along with all the test polynomials that vanish on $\mathcal{D}_n$, or at least a set of test polynomials that defines $\mathcal{D}_n$. Of course, we say ``in general'' here because it is always possible that, for example, $\mathcal{B}_n$ might have some structure that can be exploited so that showing $f_{hard,n} \in \mathcal{B}_n$ might be done without recourse to such test polynomials. However, at this level of heuristic argument, we cannot speculate on anything other than the general case.

In the general case---that is, barring some miraculous leap of ingenuity, which of course we cannot rule out---we can compare the \emph{a priori} difficulty of these approaches:
\begin{enumerate}
\item requires finding a single test polynomial $t$, verifying that $t$ vanishes on $\mathcal{C}_n$ (which implies that it vanishes on $\overline{\mathcal{C}_n} = \mathcal{A}_n$), and verifying that $t(f_{hard,n}) \neq 0$.

\item requires finding or knowing a set $t_1, \dotsc, t_k$ of test polynomials whose vanishing defines $\mathcal{B}_n$ and then verifying that $t_i(f_{hard,n}) = 0$ for all $1 \leq i \leq k$. 

\item requires constructing $\mathcal{D}_n$, along with a defining set $t_1, \dotsc, t_k$ of test polynomials, verifying that $t_i(f_{hard,n}) = 0$ for all $1 \leq i \leq k$, and verifying that $\mathcal{D}_n$ is disjoint from $\mathcal{C}_n$. 
\end{enumerate}

First, there is the obvious difference that (1) only requires finding a \emph{single} polynomial and verifying its properties, whereas both (2) and (3) require finding a whole set of polynomials and verifying their properties. Furthermore, in most such situations the number of polynomials needed in (2) and (3) will be exponential in $n$: in all the examples we are aware of except for Remark~\ref{rmk:matrices}, the sets $\mathcal{A}_n, \mathcal{B}_n, \mathcal{C}_n, \mathcal{D}_n$ have dimension $poly(n)$ and live in a space like $\Poly^{O(n)}(x_1, \dotsc, x_n)$ of dimension $2^{\Theta(n \log n)}$, which implies that any defining set of test polynomials must consist of at least $2^{\Theta(n \log n)} - poly(n) = 2^{\Theta(n \log n)}$ test polynomials.

\begin{rmk} \label{rmk:matrices}
In the case of $n \times n$ matrix multiplication the ambient space has dimension $n^6$, and in the case of matrix rigidity the ambient space has dimension $n^2$, so the above point is not an issue. However, it may be telling that even in these cases, the approach via test polynomial seems to be the most successful so far. In the case of matrix multiplication, this corresponds to border rank (see Section~\ref{sec:border}), which has been successfully used for upper bounds as well as lower bounds. In the case of matrix rigidity, see, \eg, \cite{lokamRigidity, landsbergRigidity}.
\end{rmk}

Second, we can use the complexity of the corresponding verification problems as a heuristic guide to the mathematical difficulty of the associated proofs. For starters, given a test polynomial $t$, it is easy to evaluate $t(f)$ for any explicitly given $f$. 
\begin{enumerate}
\item Verifying that $t$ vanishes on $\mathcal{C}_n$ is essentially a $\cc{coRP}$ problem. If $\mathcal{C}_n$ is the image of a simple algebraic map $\varphi$ from some $\F^N$, as most complexity classes are (see Section~\ref{sec:constructible}), we can generate random points of $\mathcal{C}_n$ by choosing random points in $\F^N$ and applying $\varphi$. In all situations we are aware of $N \leq poly(n)$.

\item Verifying that $f_{hard,n} \in \mathcal{B}_n$ requires verifying that $t_i(f_{hard,n}) = 0$ for a defining set of test polynomials $T$. We argued above that in most situations, $T$ must consist of exponentially many test polynomials.

\item Even if $\mathcal{D}_n$ is chosen to be defined by only $poly(n)$ test polynomials $t_1, \dotsc, t_{poly(n)}$---thus avoiding the difficulty of (2)---verifying that $\mathcal{D}_n$ is disjoint from $\mathcal{C}_n = \im(\varphi_n)$ reduces to deciding whether a variety given by equations is empty or not. Namely, the equations $t_i(\varphi(\var{x}))$ for $1 \leq i \leq k$, define the closed set $\varphi^{-1}(\mathcal{D}_n)$, which is empty if and only if $\mathcal{D}_n$ is disjoint from $\mathcal{C}_n$. 

Deciding whether a closed set given by equations is empty or not is the computational problem of Hilbert's Nullstellensatz ($\mathsf{HN}$), which is $\cc{NP}$-hard in general. As the $\varphi$ are quite simple, if we treat the defining equations $t_1, \dotsc, t_{poly(n)}$ as the input to our verification problem, the verification problem here is likely to be as hard as the general case of $\mathsf{HN}$. 
\end{enumerate}

Also note that the fewer test polynomials that are needed to define $\mathcal{D}_n$, the larger its dimension is, and hence the less likely it is to be disjoint from $\mathcal{C}_n$. This makes it seem unlikely that one could in fact find a $\mathcal{D}_n$ described by few test polynomials that is disjoint from $\mathcal{C}_n$ and contains $f_{hard,n}$, let alone that the corresponding instance of $\mathsf{HN}$ would not be a hard instance. Either way, we find the following complexities of the corresponding general verification problems very suggestive:

\begin{enumerate}
\item $\cc{coRP}$.

\item At least exponential time, as there are at least this many defining equations for $\mathcal{B}_n$.

\item Likely $\cc{NP}$-hard.
\end{enumerate}

Finally, in the absence of a brilliant insight to construct a $\mathcal{D}_n$ that has exponential dimension and yet is both disjoint from $\mathcal{C}_n$ and avoids the difficulty of $\mathsf{HN}$, the easiness of verification in (1) suggests that a relatively \emph{feasible computational approach} is possible using a brute force search for test modules, whereas this is not the case for approaches (2) and (3).

\section{Discussion of terminology} \label{app:terminology}
The new terminology we introduced in this paper was far from arbitrary; here we explain our reasons for choosing the terminology we did. A test $\GL_n$-module is, in particular, a representation of $\GL_n$. Indeed, the word ``module'' is often used interchangeably with ``representation'' in representation theory. In our setting, it has the additional connotation of a ``module of tests'' in the sense of computer programming. We believe the phrase ``test module'' is new.

Separating $\GL_n(\C)$-modules are essentially equivalent to the ``HWV obstructions'' of B\"{u}rgisser and Ikenmeyer \cite{BI}. In particular, the smallest $\GL_n(\C)$-module containing an HWV obstruction is a separating module, and every separating $\GL_n(\C)$-module contains some HWV obstruction (see \cite[Prop.~3.3]{BI}). We use our terminology as it generalizes (see Section~\ref{sec:general}) to other groups for which the highest weight theory does not apply, and we believe it is simpler to understand for expository purposes---in particular, it does not require knowing anything about Lie theory and the theory of highest weights. However, for certain approaches to certain lower bounds there are technical advantages to considering the highest weight vectors directly, as in \cite{BI}.

\section{Standard notation in the literature} \label{app:notation}
Rather than $\Poly^d(x_1, \dotsc, x_n)$, it is standard to see one of $\Sym^d(\C^{n})$, $\Sym^d(\C^{n*})$, $S^d(\C^n)$, or $S^d(\C^{n*})$, or even $\Sym^d(V)$ or $\Sym^d(V^*)$, or any other combination of these notations. The use of $\C^{n*}$  or $V^*$ here comes from a viewpoint in which the variables $x_i$ are viewed as the coordinate functions on an $n$-dimensional vector space $V = \C^n$, hence are elements of its (linear) dual vector space $V^* = \C^{n*}$. Sometimes the dual is dropped because it does not affect many statements. The use of $\Sym^d$ or $S^d$ is to denote the ``symmetric product'' to distinguish it from, say, the tensor product (which corresponds to noncommutative polynomials) or the wedge product (which corresponds to anti-commutative tensors, for which $x_i x_j = - x_j x_i$). 

The space of test polynomials of degree $D$ is then denoted $\Sym^D(\Sym^d(\C^n))$ (or variations similar to the above). Continuing with the viewpoint above, the coefficients $a_e$ of a polynomial $f \in \Sym^d(\C^{n*})$ are viewed as linear functions on the space of input polynomials, hence as elements of the dual vector space $\Sym^d(\C^n)$. Polynomials in the $a_e$ then live in the $D$-th symmetric power, as before.

The entire space of test polynomials is sometimes denoted $\C[\Sym^d(\C^{n*})]$ or $\mathcal{O}(\Sym^d(\C^{n*}))$; these are standard notations in algebraic geometry for the coordinate ring of the linear algebraic variety $\Sym^d(\C^n)$.

A $\GL_n$-module of type $\lambda$ is typically referred as as a Weyl module, which has several more-or-less standard notations: $V_{\lambda}$, $V_{\lambda}(\GL_n)$, $\mathbb{S}_{\lambda}(V)$ when the group is $\GL(V)$ (``$\mathbb{S}$'' for ``Schur functor''), or $\{ \lambda \}$. 

An $S_n$-module of type $\lambda$ is typically referred to as a Specht module, which also has several more-or-less standard notations, including $S_{\lambda}$ and $[\lambda]$.

In both the above cases, $\lambda$ typically refers to a partition, as the irreducible modules of $\GL_n(\C)$ are in bijective correspondence with partitions with at most $n$ parts, and the irreducible modules of $S_n$ over $\C$ are in bijective correspondence with partitions of the number $n$.

\bibliographystyle{ams-alpha}
\bibliography{gct-unity}

\end{document}